\documentclass[journal]{IEEEtran}

\usepackage{cite}
\usepackage{amsmath, amssymb, amsfonts, bm}
\usepackage{amsthm}
\usepackage{algorithmic}
\usepackage[ruled, linesnumbered]{algorithm2e}
\usepackage{setspace}
\usepackage{indentfirst}
\usepackage{graphicx}
\usepackage{epstopdf}
\usepackage{textcomp}
\usepackage{xcolor}
\usepackage{tikz}
\usetikzlibrary{arrows, shapes, chains}
\usepackage{subfigure}
\usepackage{mathrsfs}
\usepackage{verbatim}
\usepackage{url}
\usepackage{multirow}
\usepackage{multicol}
\usepackage{booktabs}
\usepackage{stfloats}
\usepackage{colortbl}
\usepackage{makecell}
\usepackage{tabularx}
\usepackage{mathtools}
\usepackage{float}

\usepackage{color}
\definecolor{colorRed}{RGB}{128, 0, 0}
\definecolor{colorGreen}{RGB}{0, 64, 0}
\definecolor{colorBlue}{RGB}{0, 0, 128}

\usepackage[colorlinks, linkcolor=colorRed, anchorcolor=blue, citecolor=colorBlue, urlcolor = colorGreen]{hyperref}

\newtheorem{assumption}{Assumption}
\newtheorem{theorem}{Theorem}
\newtheorem{remark}{Remark}
\newtheorem{definition}{Definition}

\allowdisplaybreaks[4]

\begin{document}

\title{Interference Management for Over-the-Air Federated Learning in Multi-Cell Wireless Networks}

\author{\IEEEauthorblockN{Zhibin Wang, \IEEEmembership{Graduate Student Member, IEEE}, Yong Zhou, \IEEEmembership{Member, IEEE}, \\
		Yuanming Shi, \IEEEmembership{Senior Member, IEEE}, and Weihua Zhuang, \IEEEmembership{Fellow, IEEE}}
	\thanks{
		Zhibin Wang is with the School of Information Science and Technology, ShanghaiTech University, Shanghai 201210, China, also with Shanghai Institute of Microsystem and Information Technology, Chinese Academy of Sciences, Shanghai 200050, China, and also with the University of Chinese Academy of Sciences, Beijing 100049, China (e-mail: wangzhb@shanghaitech.edu.cn).}
	\thanks{
		Yong Zhou and Yuanming Shi are with the School of Information Science and Technology, ShanghaiTech University, Shanghai 201210, China (e-mail: \{zhouyong, shiym\}@shanghaitech.edu.cn).}
	\thanks{
		Weihua Zhuang is with the Department of Electrical and Computer Engineering, University of Waterloo, Ontario N2L 3G1, Canada (e-mail: wzhuang@uwaterloo.ca).}
}

\maketitle

\begin{abstract}

Federated learning (FL) over resource-constrained wireless networks has recently attracted much attention. 
However, most existing studies consider one FL task in single-cell wireless networks and ignore the impact of downlink/uplink inter-cell interference on the learning performance. 
In this paper, we investigate FL over a multi-cell wireless network, where each cell performs a different FL task and over-the-air computation (AirComp) is adopted to enable fast uplink gradient aggregation. 
We conduct convergence analysis of AirComp-assisted FL systems, taking into account the inter-cell interference in both the downlink and uplink model/gradient transmissions, which reveals that the distorted model/gradient exchanges induce a gap to hinder the convergence of FL. 
We characterize the Pareto boundary of the error-induced gap region to quantify the learning performance trade-off among different FL tasks, based on which we formulate an optimization problem to minimize the sum of error-induced gaps in all cells. 
To tackle the coupling between the downlink and uplink transmissions as well as the coupling among multiple cells, we propose a cooperative multi-cell FL optimization framework to achieve efficient interference management for downlink and uplink transmission design. 
Results demonstrate that our proposed algorithm achieves much better average learning performance over multiple cells than non-cooperative baseline schemes.

\end{abstract}


\section{Introduction}	

The integration of artificial intelligence (AI) and wireless communications becomes an emerging trend for accelerating the fulfillment of connected intelligence in the 6G wireless networks \cite{letaief2019roadmap, yang2020artificial, shi2021mobile}.
Meanwhile, a vast amount of raw data generated at the wireless edge stimulates the emergence of various AI applications, such as smart Internet of Things (IoT), autonomous driving, and augmented reality.
Due to the communication-efficiency and privacy-preserving concerns, it is undesirable to transfer raw data from massive edge devices to an edge server, e.g., base station (BS), as in the centralized machine learning (ML) framework.
To this end, federated learning (FL) as one of the promising distributed learning paradigms has recently been proposed to train a common ML model in an iterative manner by exploiting the edge computing power while keeping the raw data at the edge devices \cite{konevcny2016federated, mcmahan2017communication}.
During each training iteration, multiple edge devices first receive the latest global model broadcast from the edge server, based on which the edge devices compute the local updates with their own datasets. 
Subsequently, the edge server aggregates the local updates to refresh the global model for the next training iteration.
The aggregation of local updates therein can be regarded as an implicit information sharing without transmitting the private raw data, which effectively preserves the data privacy.
However, because of the periodic exchange of high-dimensional model/gradient parameters between the edge server and edge devices during the training process, it is mandatory to develop communication-efficient strategies to reduce the communication overhead and accelerate the training process, especially when the radio resources are limited.

Recent years have witnessed an upsurge research interest in the joint learning and communication design for FL over wireless networks.
One challenge in realizing such a joint design is the efficient allocation of limited radio resources for high learning performance.
To this end, the authors in \cite{wang2019adaptive} propose to dynamically adapt the frequency of global aggregation to improve the learning performance under a fixed resource budget for the entire training process.
The authors in \cite{dinh2021federated} optimize the resource allocation to balance the trade-off between the training time and the energy consumption.
In addition, various scheduling policies are developed in \cite{yang2020scheduling, kang2020reliable, amiri2021convergence, ren2020scheduling, chen2021joint, shi2021joint, wei2022lowlatency, xu2021client} to allocate limited radio bandwidth to a subset of edge devices with good channel conditions \cite{yang2020scheduling, amiri2021convergence}, large local updates norm \cite{amiri2021convergence, ren2020scheduling}, small transmission delay \cite{chen2021joint, shi2021joint, wei2022lowlatency}, and low energy consumption \cite{xu2021client}.
Moreover, hierarchical FL is considered in \cite{abad2020hierarchical, luo2020hfel, lim2021dynamic, lim2022decentralized} to reduce communication costs and achieve efficient resource allocation.
All the aforementioned studies adopt an orthogonal multiple access (OMA) scheme to achieve uplink model/gradient aggregation, where the edge sever updates the global model after successfully decoding each of the local updates.
The number of resource blocks required by such a ``\textit{transmit-then-aggregate}'' strategy linearly scales with the number of participating devices, which can be radio spectrum inefficient, especially when the number of devices is large.

Over-the-air computation (AirComp) has recently been leveraged to enhance the communication efficiency of uplink model/gradient aggregation in FL systems \cite{yang2020federated, zhu2020broadband, amiri2020machine, fang2022communicationefficient}.
By exploiting the waveform superposition property of multiple-access channels, AirComp allows multiple devices to concurrently transmit their local updates over the same radio channel and meanwhile enables the edge server to directly receive an aggregation of these local updates \cite{zhu2021aircomp, wang2021wirelesspowered, fang2021overtheair}.
Such an ``\textit{aggregate-when-transmit}'' strategy requires only one resource block regardless of the number of participating devices.
To achieve accurate model aggregation via AirComp, various transceiver designs have recently been proposed to reduce the distortion of the aggregated model due to the receiver noise and channel fading, such as via multi-antenna beamforming \cite{yang2020federated} and transmit power control \cite{zhang2021gradient, cao2021optimized}.
In addition, auxiliary equipment is leveraged to support reliable uplink model/gradient aggregation, such as intelligent reflecting surface \cite{yang2020fedirs, liu2021reconfigurable, wang2022federated} and half-duplex relay \cite{lin2022relayassisted}, thereby improving the learning performance when the propagation environment is unfavorable.
All the aforementioned works assume ideal error-free downlink transmissions and focus on the uplink transmission design.
Due to the random channel fading, the global model disseminated in the downlink inevitably suffers from channel distortion, which in turn detrimentally affects the local training performance at the edge devices.
When the downlink distortion is severe, the previous training may become futile.

To fill this gap, the authors in \cite{amiri2021noisydl, wei2022federated, wang2022edge} develop efficient methods to tackle non-ideal downlink model dissemination in wireless FL systems.
In particular, the authors in \cite{amiri2021noisydl} show that, by adopting the analog downlink transmission, the adverse impact on the learning performance caused by the devices with less accurate model estimates can be alleviated by those with more accurate model estimates.
The convergence of FL over noisy downlink and uplink wireless channels is investigated in \cite{wei2022federated}, which takes into account full and partial device participation, uplink model and model differential transmission, and non-independent and identically distributed (non-i.i.d.) local datasets.
Results show that the effect of both downlink and uplink communication errors can be alleviated during the training process as long as they do not dominate the errors caused by stochastic gradient descent (SGD).
In addition, the authors in \cite{wang2022edge} consider a unit-modulus AirComp-assisted FL system, where a phase-shift network at the edge server is utilized to reduce the parameter distortion caused by both the downlink and uplink communication errors.

\subsection{Motivations}

The aforementioned studies develop efficient joint learning and communication designs considering one FL task in a single-cell wireless network.
With the rapid advancement in connected intelligence, the co-existence of multiple FL tasks in a multi-cell wireless network will become a new normal in future wireless networks.
Under this circumstance, the inter-cell interference among different FL tasks needs to be considered, as it is the main performance-limiting factor of multi-cell wireless networks.
A few recent studies have given some preliminary discussions about the impact of inter-cell interference on device scheduling and resource allocation in wireless FL systems \cite{salehi2021federated, lin2021deploying}.
These studies focus on inter-cell interference in the uplink aggregation.
Different from the existing studies in the literature, in this paper we study multiple FL tasks over multi-cell wireless networks, while taking into account the inter-cell interference in both the downlink and uplink model/gradient transmissions.
Furthermore, we investigate the learning performance in multiple cells, extending the existing studies focusing only on a typical cell.

\subsection{Challenges and Contributions}

In this paper, we consider over-the-air FL in a multi-cell wireless network, where each cell performs a different FL task and the devices in each cell upload local gradients to their home BS using AirComp.
We consider a practical yet challenging scenario with universal frequency reuse, where both the downlink and uplink transmissions during the training process in each cell are distorted by the receiver noise, channel fading, and inter-cell interference.
To achieve maximal learning performance for all cells, it is necessary to effectively manage inter-cell interference during both the downlink and uplink model/gradient transmissions, which introduces the following challenges.
First, the main objective of FL is to enable low-latency and high-quality intelligence distillation from edge devices. However, there does not exist a widely-accepted performance metric that can characterize the learning performance as a function of communication system parameters.
Second, considering both the downlink and uplink model/gradient distortion significantly complicates the learning performance characterization.
This is because the local gradient computation depends on the distortion of the downlink model transmission, while the global model update also depends on the distortion of the uplink gradient transmission.
Third, due to the inter-cell interference, the transmission processes in different cells are coupled and hence simply focusing on single-cell optimization may deteriorate the learning performance in other cells, which makes it challenging to coordinate the downlink and uplink model/gradient transmissions in different cells.
To address these challenges, we develop a cooperative optimization framework to balance the learning performance of different FL tasks in a multi-cell wireless network.
The main contributions of this paper are summarized as follows:
\begin{itemize}
	\item
	We first analyze the convergence of the proposed over-the-air FL in a multi-cell wireless network, taking into account the impact of both the downlink and uplink transmission distortions caused by the receiver noise, channel fading, and inter-cell interference.
	The derived analytical convergence expression shows that both the downlink and uplink transmission distortions induce a gap that prevents the FL algorithm from converging to a stationary point;
	
	\item
	To balance the learning performance among different FL tasks in multiple cells, we define a new performance metric, named gap region, to be the set of error-induced gaps that can be simultaneously generated by all cells under resource constraints.
	We further introduce the Pareto boundary of the gap region to characterize the performance trade-off among multiple cells, which enables us to formulate an optimization problem to minimize the sum of error-induced gaps for all cells via the profiling technique;
	
	\item
	We propose a cooperative multi-cell FL optimization framework, which decouples the optimization for the downlink and uplink model/gradient transmissions, and coordinates the optimization for different cells.
	It enables us to separately minimize the sum of gaps generated by the downlink and uplink transmission distortions, where the resulting subproblems can be solved by using the bisection search method via checking the feasibility of a sequence of second order cone programming (SOCP) problems.
	Practical implementation issues of our proposed optimization framework are discussed;
	
	\item
	We present extensive simulations to validate our proposed cooperative multi-cell FL optimization framework.
	Simulation results show that our proposed algorithm can balance the learning performance among different FL tasks in multiple cells, and yield lower training loss and higher test accuracy as compared with non-cooperative baseline schemes in terms of the multi-cell average learning performance.
\end{itemize}

\subsection{Organization and Notations}

The rest of this paper is organized as follows.
Section \ref{Sec:SysModel} describes the learning and communication models for over-the-air FL in a multi-cell wireless network.
Section \ref{Sec:CvgPrb} presents the convergence analysis and formulates the cooperative optimization problem.
In Section \ref{Sec:SysOpt}, we propose a cooperative multi-cell FL optimization framework to achieve efficient interference management.
Simulation results are provided in Section \ref{Sec:Sim} to evaluate the performance of our proposed algorithm.
Finally, Section \ref{Sec:Conc} concludes this work.

\textit{Notations:}
Italic, bold lower-case, and bold upper-case letters denote scalar, column vector, and matrix, respectively.
Operators $(\cdot)^\dagger$, $(\cdot)^{\sf T}$, and ${\rm diag}(\cdot)$ denote the conjugate, transpose, and diagonal matrix, respectively.
Besides, operator $|\cdot|$ returns the cardinality of a set or the absolute value of a scalar, and $\|\cdot\|$ represents the Euclidean norm, while $\Re\{\cdot\}$ denotes the real part of a complex value.
The important notations used throughout the paper are listed in Table \ref{tab:notation}.

\begin{table}[t]
	\centering
	\caption{Important notations and their definitions.}
	\label{tab:notation}
	\footnotesize
	\renewcommand{\arraystretch}{1.3}
	\begin{tabular}{!{\vrule width1pt}l!{\vrule width1pt}p{6.3cm}!{\vrule width1pt}}
		\Xhline{1pt}
		\rowcolor[HTML]{EFEFEF}
		\textbf{Notation} & \textbf{Definition} \\
		\Xhline{1pt}
		$\mathcal{M}$ & Set of all cells/BSs. \\ \hline
		$\mathcal{K}_m$ & Set of all devices associated with BS $m \in \mathcal{M}$. \\ \hline
		$\mathcal{D}_k$ & Local dataset at device $k \in \mathcal{K}_m$. \\ \hline
		$\bm{w}_m$ & Model parameter of the learning task in cell $m \in \mathcal{M}$. \\ \hline
		$f_m(\bm{w}_m; \bm{\xi}_i)$ & Sample-wise loss function of the learning task in cell $m$ evaluated at model parameter $\bm{w}_m$ on sample $\bm{\xi}_i$. \\ \hline
		$F_{m, k}(\bm{w}_m)$ & Local loss function of device $k \in \mathcal{K}_m$ evaluated at model parameter $\bm{w}_m$ on dataset $\mathcal{D}_k$. \\ \hline
		$F_m(\bm{w}_m)$ & Global loss function evaluated at model parameter $\bm{w}_m$ for the learning task in cell $m$. \\ \hline
		$T$ & Number of communication rounds. \\ \hline
		$\bm{g}_k$ & Local gradient parameter at device $k \in \mathcal{K}_m$. \\ \hline
		$\bm{s}_m^{\rm dl}$ / $\bm{s}_k^{\rm ul}$ & Normalized \textit{model parameter of $\bm{w}_m$} / \textit{gradient parameter $\bm{g}_k$}. \\ \hline
		$h_k^{\rm dl}$ / $h_k^{\rm ul}$ & \textit{Downlink} / \textit{Uplink} channel coefficient between device $k \in \mathcal{K}_m$ and its home BS $m \in \mathcal{M}$. \\ \hline
		$h_{l, k}^{\rm dl}$ / $h_{k, l}^{\rm ul}$ & \textit{Downlink} / \textit{Uplink} channel coefficient between device $k \in \mathcal{K}_m$ and its non-associated BS $l \in \mathcal{M} \setminus \{m\}$. \\ \hline
		$y_{k, d}^{\rm dl}$ / $y_{m, d}^{\rm ul}$ & \textit{Downlink} / \textit{Uplink} signal of the $d$-th dimension received at \textit{device $k \in \mathcal{K}_m$} / \textit{BS $m \in \mathcal{M}$}. \\ \hline
		$p_m^{\rm dl}$ / $p_k^{\rm ul}$ & \textit{Downlink} / \textit{Uplink} transmit power at \textit{BS $m \in \mathcal{M}$} / \textit{device $k \in \mathcal{K}_m$}. \\ \hline
		$z_{k, d}^{\rm dl}$ / $z_{m, d}^{\rm ul}$ & Additive receiver noise of the $d$-th dimension at \textit{device $k \in \mathcal{K}_m$} / \textit{BS $m \in \mathcal{M}$}. \\ \hline
		$\bm{r}_k^{\rm dl}$ / $\bm{r}_m^{\rm ul}$ & De-normalized signal at \textit{device $k \in \mathcal{K}_m$} / \textit{BS $m \in \mathcal{M}$}. \\ \hline
		$\bm{e}_k^{\rm dl}$ / $\bm{e}_m^{\rm ul}$ & \textit{Downlink dissemination} / \textit{Uplink aggregation} error at \textit{device $k \in \mathcal{K}_m$} / \textit{BS $m \in \mathcal{M}$}. \\ \hline
		$\widehat{\bm{w}}_k$ & Model parameter received at device $k \in \mathcal{K}_m$. \\ \hline
		$\widehat{\bm{g}}_m$ & Average of the local gradients received at BS $m \in \mathcal{M}$. \\ \hline
		$c_m$ & Receive normalizing factor at BS $m$. \\ \hline
		$\eta_m$ & Learning rate of the learning task in cell $m \in \mathcal{M}$. \\ 
		\Xhline{1pt}
	\end{tabular}
\end{table}

\section{System Model} \label{Sec:SysModel}

In this section, we present the learning and communication models for over-the-air FL in a multi-cell wireless network, taking into account the inter-cell interference in both the downlink model dissemination and uplink gradient aggregation.

\subsection{Learning Model}

Consider a multi-cell wireless network consisting of $M$ single-antenna BSs, where BS $m \in \mathcal{M} = \{1, 2, \dots, M\}$ aims to train an ML model by coordinating $K_m$ single-antenna devices located in cell $m$, as shown in Fig. \ref{fig:system}.
Specifically, device $k \in \mathcal{K}_m = \{\sum_{l = 1}^{m - 1} K_l + 1, \sum_{l = 1}^{m - 1} K_l + 2, \dots, \sum_{l = 1}^{m - 1} K_l + K_m\}$ is associated with BS $m$ and has local dataset $\mathcal{D}_k$, where $\mathcal{K}_m \cap \mathcal{K}_l = \emptyset$, $\forall \, l \in \mathcal{M} \setminus \{m\}$.
The local loss function of device $k \in \mathcal{K}_m$ evaluated at model parameter $\bm{w}_m \in \mathbb{R}^{D_m}$ is given by
\begin{align} \label{eq:local_loss}
	F_{m, k}(\bm{w}_m) = \frac{1}{|\mathcal{D}_k|} \sum_{\bm{\xi}_i \in \mathcal{D}_k} f_m(\bm{w}_m; \bm{\xi}_i)
\end{align}
where $\bm{\xi}_i$ denotes the $i$-th labeled data sample, and $f_m(\bm{w}_m; \bm{\xi}_i)$ is the sample-wise loss function quantifying the prediction error of model parameter $\bm{w}_m$ on data sample $\bm{\xi}_i$.
Without loss of generality, all the local datasets in the same cell are assumed to have the same size, i.e., $|\mathcal{D}_k| = |\mathcal{D}_j|$, $\forall \, k, j \in \mathcal{K}_m$.
Accordingly, the global loss function for the learning task in cell $m$ can be expressed as
\begin{align} \label{eq:global_loss}
	F_m(\bm{w}_m) = \frac{1}{K_m} \sum_{k \in \mathcal{K}_m} F_{m, k}(\bm{w}_m)
\end{align}
and the objective of FL is to obtain a global model such that
\begin{align} \label{eq:goal}
	\bm{w}_m^\star = {\rm arg} \underset{\bm{w}_m \in \mathbb{R}^{D_m}}{\rm min} \, F_m(\bm{w}_m).
\end{align}

\begin{figure}[t]
	\centering
	\includegraphics[scale=0.31]{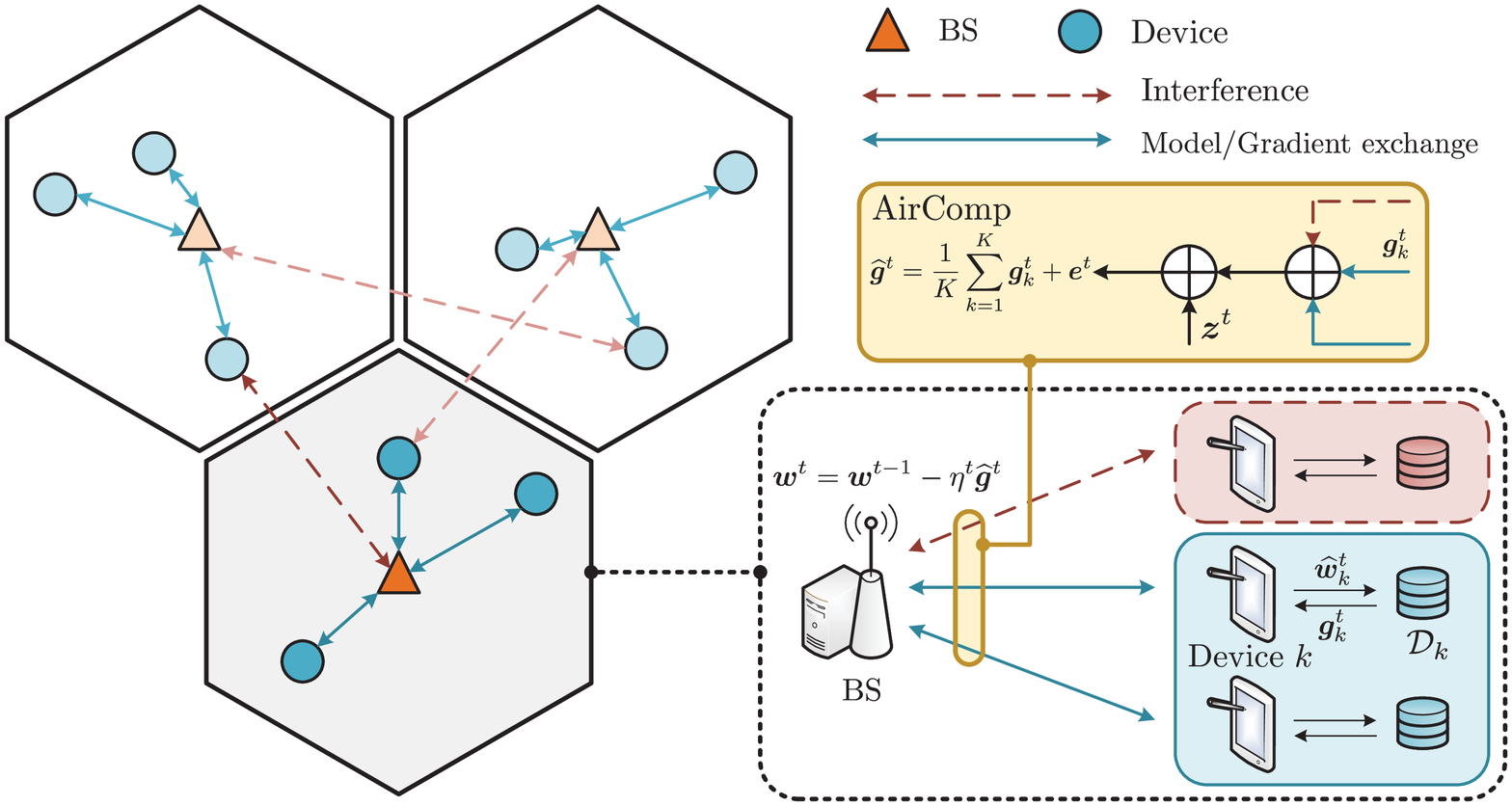}
	\caption{Illustration of multiple FL tasks in a multi-cell wireless network.} 
	\label{fig:system}
\end{figure}

In our considered multi-cell system, each cell performs a different FL task by applying the gradient-averaging method \cite{konevcny2016federated}, which will be elaborated in the next subsection.
We consider universal frequency reuse, i.e., all cells share the same frequency channel, which inevitably leads to inter-cell interference among the different cells.

\subsection{Communication Model}

To train a high-quality global model, each cell iteratively performs the training process for $T \in \mathbb{N}_+$ rounds. 
Each training round consists of four stages: 1) \textit{downlink model dissemination}, 2) \textit{local gradient computation}, 3) \textit{uplink gradient aggregation}, and 4) \textit{global model update}, where stages 1) and 3) are carried out over wireless fading channels.
For simplicity, we assume that the sample-wise loss functions and the dimensions of model parameters in different cells are the same, i.e., $f_m(\cdot) = f(\cdot)$ and $D_m = D$, $\forall \, m \in \mathcal{M}$.
Note that for a general case that the dimensions of model parameters in different cells are different, we can use either zero padding or compression technique to adjust the model parameters of different cells to have the same size.
The details of the learning and communication design for the $t$-th round, $t \in \mathcal{T} = \{1, 2, \dots, T\}$, are given in the following.

\subsubsection{Downlink Model Dissemination}

In the downlink, BS $m$ broadcasts its current global model $\bm{w}_m^{t - 1}$ to the associated devices in set $\mathcal{K}_m$.
Specifically, at the $t$-th round, to facilitate the transmit power control, BS $m$ first normalizes the current global model $\bm{w}_m^{t - 1}$ as
\begin{align}
	(\bm{s}_m^{\rm dl})^t = \frac{\bm{w}_m^{t - 1} - \bar{w}_m^{t - 1} \bm{1}}{\nu_m^{t - 1}}
\end{align}
where $\bm{1} = [1, 1, \dots, 1]^{\sf T} \in \mathbb{R}^D$, and $\bar{w}_m^{t - 1} \in \mathbb{R}$ and $\nu_m^{t - 1} \in \mathbb{R}_{++}$ denote the mean and standard deviation of $D$ entries of global model $\bm{w}_m^{t - 1}$, respectively.
They are defined as
\begin{align}
	\bar{w}_m^{t - 1} = \frac{1}{D} \sum_{d = 1}^{D} w_{m, d}^{t - 1}, \, (\nu_m^{t - 1})^2 = \frac{1}{D} \sum_{d = 1}^{D} \left(w_{m, d}^{t - 1} - \bar{w}_m^{t - 1}\right)^2.
\end{align}
Hence, each element of normalized model $(\bm{s}_m^{\rm dl})^t$ has zero mean and unit variance, i.e., $\mathbb{E}[(s_{m, d}^{\rm dl})^t] = 0$ and $\mathbb{E}[((s_{m, d}^{\rm dl})^t)^2] = 1$, $\forall \, d \in \{1, 2, \dots, D\}$.
For independent FL tasks in different cells, we have $\mathbb{E}[(\bm{s}_m^{\rm dl})^t ((\bm{s}_l^{\rm dl})^t)^{\sf T}] = \bm{0}$, $\forall \, m \ne l$.

Let $(h_k^{\rm dl})^t \in \mathbb{C}$ and $(h_{l, k}^{\rm dl})^t \in \mathbb{C}$, $\forall \, k \in \mathcal{K}_m$, denote the downlink channel coefficients during the $t$-th round between device $k$ and its home BS $m$, and between device $k$ and its non-associated BS $l \in \mathcal{M} \setminus \{m\}$, respectively.
Assume that each device is able to accurately estimate the channels between itself and all the BSs through downlink pilot signaling, and each BS can obtain the local channel state information (CSI) of its associated devices through uplink feedback \cite{cao2021cooperative}.
The channel gain of each link is invariant within one coherence time block but independently changes from one block to another.
Each round consists of one downlink transmission block and one uplink transmission block, where each transmission block can accommodate the transmission of an entire model/gradient parameter.
The communication protocol designed in the following focuses on the transmission of one typical dimension of the parameters, which can be extended to the scenarios where multiple blocks are required for transmitting high-dimensional model/gradient parameters.
Note that, when the dimension of the model/gradient parameter is quite large, the compression techniques can be leveraged to enable the parameters to be transmitted in one coherence block \cite{amiri2020machine, amiri2020federated}.

BS $m$ broadcasts the normalized global model, $(\bm{s}_m^{\rm dl})^t$, to the associated devices in set $\mathcal{K}_m$.
The signal received at device $k \in \mathcal{K}_m$ is given by
\begin{align}
	(y_{k, d}^{\rm dl})^t & = (h_k^{\rm dl})^t \sqrt{(p_m^{\rm dl})^t} (s_{m, d}^{\rm dl})^t \notag \\
	& \quad + \underset{\text{Inter-cell interference}}{\underbrace{\sum_{l \in \mathcal{M} \setminus \{m\}} (h_{l, k}^{\rm dl})^t \sqrt{(p_l^{\rm dl})^t} (s_{l, d}^{\rm dl})^t}} + (z_{k, d}^{\rm dl})^t
\end{align}
where $(p_m^{\rm dl})^t \in \mathbb{R}_+$ denotes the downlink transmit power at BS $m$ and $(z_{k, d}^{\rm dl})^t \sim \mathcal{CN}(0, (\sigma_k^{\rm dl})^2)$ denotes the additive receiver Gaussian noise at device $k$.
By assuming that BS $m$ shares small-sized scalars $\{\bar{w}_m^{t - 1}, \nu_m^{t - 1}\}$ and $(p_m^{\rm dl})^t$ with its associated devices in an error-free manner before broadcasting the global model, device $k$ is able to de-normalize the received signal as follows
\begin{align} \label{eq:r_dl_k_d}
	(r_{k, d}^{\rm dl})^t & = \frac{((h_k^{\rm dl})^t)^\dagger \nu_m^{t - 1}}{|(h_k^{\rm dl})^t|^2 \sqrt{(p_m^{\rm dl})^t}} (y_{k, d}^{\rm dl})^t + \bar{w}_m^{t - 1} \notag \\
	& = w_{m, d}^{t - 1} + \underbrace{\frac{((h_k^{\rm dl})^t)^\dagger \nu_m^{t - 1}}{|(h_k^{\rm dl})^t|^2 \sqrt{(p_m^{\rm dl})^t}}} \notag \\
	& \quad \underset{(e_{k, d}^{\rm dl})^t}{\underbrace{\times \left(\sum_{l \in \mathcal{M} \setminus \{m\}} (h_{l, k}^{\rm dl})^t \sqrt{(p_l^{\rm dl})^t} (s_{l, d}^{\rm dl})^t + (z_{k, d}^{\rm dl})^t\right)}}
\end{align}
where $(e_{k, d}^{\rm dl})^t$ denotes the downlink dissemination error caused by the receiver noise, channel fading, and inter-cell interference.
Here, the first term in parentheses of $(e_{k, d}^{\rm dl})^t$ denotes the error caused by inter-cell interference and the second term is the error induced by the additive noise.
Therefore, the model parameter received at device $k$ can be represented as
\begin{align} \label{eq:dl_w}
	\widehat{\bm{w}}_k^t = \Re\left\{(\bm{r}_k^{\rm dl})^t\right\} = \bm{w}_m^{t - 1} + \Re\left\{(\bm{e}_{k}^{\rm dl})^t\right\}
\end{align}
where $(\bm{r}_k^{\rm dl})^t = [(r_{k, 1}^{\rm dl})^t, (r_{k, 2}^{\rm dl})^t, \dots, (r_{k, D}^{\rm dl})^t]^{\sf T}$ and $(\bm{e}_k^{\rm dl})^t = [(e_{k, 1}^{\rm dl})^t, (e_{k, 2}^{\rm dl})^t, \dots, (e_{k, D}^{\rm dl})^t]^{\sf T}$.

\subsubsection{Local Gradient Computation}

After obtaining estimated global model $\widehat{\bm{w}}_k^t$, device $k$ computes local gradient $\bm{g}_k^t \in \mathbb{R}^{D}$ based on its local dataset $\mathcal{D}_k$, given by
\begin{align} \label{eq:local_grad}
	\bm{g}_k^t = \nabla F_k(\widehat{\bm{w}}_k^t) = \frac{1}{|\mathcal{D}_k|} \sum_{\bm{\xi}_i \in \mathcal{D}_k} \nabla f(\widehat{\bm{w}}_k^t; \bm{\xi}_i).
\end{align}
Note that the mini-batch stochastic gradient can be adopted when the size of the local dataset is large.
By uniformly sampling a batch of data $(\mathcal{D}_k^{\rm bs})^t \subset \mathcal{D}_k$, the mini-batch stochastic gradient can be computed by
\begin{align}
	\widetilde{\bm{g}}_k^t = \frac{1}{|(\mathcal{D}_k^{\rm bs})^t|} \sum_{\bm{\xi}_i \in (\mathcal{D}_k^{\rm bs})^t} \nabla f(\widehat{\bm{w}}_k^t; \bm{\xi}_i)
\end{align}
where $\mathbb{E}_{(\mathcal{D}_k^{\rm bs})^t}\left[\widetilde{\bm{g}}_k^t\right] = \bm{g}_k^t$.

In the following, we mainly adopt \eqref{eq:local_grad} for the convergence analysis and device $k \in \mathcal{K}_m$ uploads its local gradient parameter $\bm{g}_k^t$ to its home BS $m$ in the uplink transmission.

\subsubsection{Uplink Gradient Aggregation}

For global model update, BS $m$ aims to receive an arithmetic mean of the local gradients of all devices in set $\mathcal{K}_m$, given by
\begin{align} \label{eq:target}
	\bm{g}_m^t = \frac{1}{K_m} \sum_{k \in \mathcal{K}_m} \bm{g}_k^t.
\end{align}
%This can be achieved by using the existing OMA schemes, e.g., orthogonal frequency division multiple access (OFDMA).
%In particular, BS $m$ first receives each of the local updates from all devices in set $\mathcal{K}_m$ and then computes their arithmetic mean as in \eqref{eq:target}.
%However, with such a ``\textit{transmit-then-aggregate}" strategy, the number of required resource blocks linearly scales with the number of devices. 
%This strategy can be radio spectrum inefficient as each BS is interested in receiving only the arithmetic mean of the local gradients, rather than each local update.
To enhance the communication efficiency, we adopt AirComp for uplink gradient aggregation.
%By exploiting the superposition property of multiple-access channels, AirComp seamlessly integrates communication and computation, and can be regarded as an ``\textit{aggregate-when-transmit}" strategy.
With AirComp, BS $m$ is capable of directly obtaining a noisy version of the arithmetic mean in \eqref{eq:target} by allowing all devices in set $\mathcal{K}_m$ to concurrently transmit their local gradients over the same radio channel \cite{yang2020federated, zhu2020broadband, amiri2020machine}.
%Besides, inherent privacy protection can be provided by AirComp via hiding individual gradient information in the superimposed signals masked by the receiver noise \cite{zhu2021aircomp, liu2021privacy, elgabli2021harnessing}.
Similar to the downlink transmission, device $k$ first normalizes its local gradient as \cite{liu2021reconfigurable}
\begin{align} \label{eq:normalize_grad}
	(\bm{s}_k^{\rm ul})^t = \frac{\bm{g}_k^t - \bar{g}_k^t \bm{1}}{\upsilon_k^t}
\end{align}
where $\bar{g}_k^t \in \mathbb{R}$ and $\upsilon_k^t \in \mathbb{R}_{++}$ denote the mean and standard deviation of $D$ entries of local gradient $\bm{g}_k^t$, respectively.
They are defined as
\begin{align}
	\bar{g}_k^t = \frac{1}{D} \sum_{d = 1}^{D} g_{k, d}^t, \, (\upsilon_k^t)^2 = \frac{1}{D} \sum_{d = 1}^{D} \left(g_{k, d}^t - \bar{g}_k^t\right)^2.
\end{align}
Consequently, \eqref{eq:normalize_grad} ensures that the transmitted signals satisfy $\mathbb{E}[(s_{k, d}^{\rm ul})^t] = 0$ and $\mathbb{E}[((s_{k, d}^{\rm ul})^t)^2] = 1$, $\forall \, d \in \{1, 2, \dots, D\}$, $\forall \, k \in \mathcal{K}_m$.
By assuming independent local gradients among different devices, we have $\mathbb{E}[(\bm{s}_k^{\rm ul})^t ((\bm{s}_{k'}^{\rm ul})^t)^{\sf T}] = \bm{0}$, $\forall \, k \ne k'$.

Let $(h_k^{\rm ul})^t \in \mathbb{C}$ and $(h_{k, l}^{\rm ul})^t \in \mathbb{C}$, $\forall \, k \in \mathcal{K}_m$, denote the uplink channel coefficients during the $t$-th round between device $k$ and its home BS $m$, and between device $k$ and its non-associated BS $l \in \mathcal{M} \setminus \{m\}$, respectively.
With concurrent transmissions in the uplink, the signal received at BS $m$ is given by
\begin{align} \label{eq:y_ul}
	(y_{m, d}^{\rm ul})^t & = \sum_{k \in \mathcal{K}_m} (h_k^{\rm ul})^t b_k^t (s_{k, d}^{\rm ul})^t + (z_{m, d}^{\rm ul})^t \notag \\
	& \quad + \underset{\text{Inter-cell interference}}{\underbrace{\sum_{l \in \mathcal{M} \setminus \{m\}} \sum_{k' \in \mathcal{K}_l} (h_{k', m}^{\rm ul})^t b_{k'}^t (s_{k', d}^{\rm ul})^t}}
\end{align}
where $b_k^t \in \mathbb{C}$ denotes the transmit scalar at device $k$ and $(z_{m, d}^{\rm ul})^t \sim \mathcal{CN}(0, (\sigma_m^{\rm ul})^2)$ denotes the additive receiver Gaussian noise at BS $m$.
To compensate for the phase distortion introduced by complex channel responses, the transmit scalar is set to $b_k^t = \frac{((h_k^{\rm ul})^t)^\dag}{|(h_k^{\rm ul})^t|} \sqrt{(p_k^{\rm ul})^t}$, where $(p_k^{\rm ul})^t \in \mathbb{R}_+$ denotes the uplink transmit power at device $k$.
Then, \eqref{eq:y_ul} reduces to
\begin{align}
	(y^{\rm ul}_{m, d})^t & = \sum_{k \in \mathcal{K}_m} |(h_k^{\rm ul})^t| \sqrt{(p_k^{\rm ul})^t} (s_{k, d}^{\rm ul})^t + (z_{m, d}^{\rm ul})^t \notag \\*
	& \quad + \sum_{l \in \mathcal{M} \setminus \{m\}} \sum_{k' \in \mathcal{K}_l} \frac{(h_{k', m}^{\rm ul})^t ((h_{k'}^{\rm ul})^t)^\dag}{|(h_{k'}^{\rm ul})^t|} \sqrt{(p_{k'}^{\rm ul})^t} (s_{k', d}^{\rm ul})^t.
\end{align}
By assuming that the devices report scalars $\{\bar{g}_k^t, \upsilon_k^t\}_{k \in \mathcal{K}_m}$ to their home BSs in an error-free manner, BS $m$ is able to obtain the scaled signal as follows
\begin{align} \label{eq:r_ul_m_d}
	& (r_{m, d}^{\rm ul})^t = \frac{1}{K_m} \left(\frac{(y_{m, d}^{\rm ul})^t}{\sqrt{c_m^t}} + \sum_{k \in \mathcal{K}_m} \bar{g}_k^t\right) \notag \\
	& = \frac{1}{K_m} \left(\sum_{k \in \mathcal{K}_m} g_{k, d}^t + \frac{(y_{m, d}^{\rm ul})^t}{\sqrt{c_m^t}} - \sum_{k \in \mathcal{K}_m} \left(g_{k, d}^t - \bar{g}_k^t\right)\right) \notag \\
	& = \frac{1}{K_m} \sum_{k \in \mathcal{K}_m} g_{k, d}^t + \frac{1}{K_m} \notag \\
	& \quad \times \underbrace{\left(\sum_{k \in \mathcal{K}_m} \left(\frac{|(h_k^{\rm ul})^t| \sqrt{(p_k^{\rm ul})^t}}{\sqrt{c_m^t}} - \upsilon_k^t \right) (s_{k, d}^{\rm ul})^t + \frac{(z_{m, d}^{\rm ul})^t}{\sqrt{c_m^t}}\right.} \notag \\
	& \quad \underset{(e_{m, d}^{\rm ul})^t}{\underbrace{\left. + \sum_{l \in \mathcal{M} \setminus \{m\}} \sum_{k' \in \mathcal{K}_l} \frac{(h_{k', m}^{\rm ul})^t ((h_{k'}^{\rm ul})^t)^\dagger \sqrt{(p_{k'}^{\rm ul})^t} (s_{k', d}^{\rm ul})^t}{|(h_{k'}^{\rm ul})^t| \sqrt{c_m^t}}\right)}}
\end{align}
where $c_m^t \in \mathbb{R}_{++}$ denotes the receive normalizing factor for signal power alignment and noise power suppression, and $(e_{m, d}^{\rm ul})^t$ denotes the uplink aggregation error caused by the receiver noise, channel fading, and inter-cell interference.
Here, the first term of $(e_{m, d}^{\rm ul})^t$ represents the misalignment error caused by the non-uniform channel fading and power control, the second term denotes the error introduced by the receiver noise, and the third term denotes the error due to inter-cell interference.
Therefore, the average of the local gradients received at BS $m$ is given by
\begin{align} \label{eq:ul_g}
	\widehat{\bm{g}}_m^t = \Re\left\{(\bm{r}_m^{\rm ul})^t\right\} = \frac{1}{K_m} \left(\sum_{k \in \mathcal{K}_m} \bm{g}_k^t + \Re\left\{(\bm{e}_m^{\rm ul})^t\right\}\right)
\end{align}
where $(\bm{r}_m^{\rm ul})^t = [(r_{m, 1}^{\rm ul})^t, (r_{m, 2}^{\rm ul})^t, \dots, (r_{m, D}^{\rm ul})^t]^{\sf T}$ and $(\bm{e}_m^{\rm ul})^t = [(e_{m, 1}^{\rm ul})^t, (e_{m, 2}^{\rm ul})^t, \dots, (e_{m, D}^{\rm ul})^t]^{\sf T}$.

\subsubsection{Global Model Update}

Based on the average of local gradients received at BS $m$ given in \eqref{eq:ul_g}, the global model maintained by BS $m$ can be updated as
\begin{align} \label{eq:global_update}
	\bm{w}^t_m = \bm{w}_m^{t - 1} - \eta_m^t \widehat{\bm{g}}_m^t
\end{align}
where $\eta_m^t$ denotes the learning rate of cell $m$.

\begin{remark}
	In the AirComp-based system, synchronization is required among the distributed devices.
	In practice, synchronization can be realized by sharing a reference-clock across the devices \cite{abari2015airshare}, or adopting the timing advance technique commonly used in 4G Long-Term Evolution (LTE) and 5G New Radio (NR) \cite{mahmood2019time}.
	For the multi-cell system under consideration, we assume that all devices can be synchronized as in \cite{cao2021cooperative}.
\end{remark}

\begin{remark}
		Normalizing the model/gradient parameters before transmission has two benefits.
		First, by normalizing the parameters to have zero-mean entries, the model/gradient parameters obtained by \eqref{eq:dl_w} and \eqref{eq:ul_g} can be regarded as unbiased estimates of their original values, thereby facilitating the convergence analysis.
		Second, by normalizing the parameters to have unit-variance entries, we need to focus only on the impact of the power control of other cells on the error-induced gap regardless of the specific values of their model/gradient parameters.
\end{remark}

\begin{remark}
	As mentioned above, compression techniques can be used to adjust the model sizes in different cells to be the same and enhance the communication efficiency.
	If compression is applied, then the effect of compression errors should be considered in addition to the wireless communication errors.
\end{remark}

\begin{remark}
	Multiple antennas can be equipped at both the receiver and transmitter to improve the transmission accuracy by exploiting the diversity gain, and to simultaneously transmit multiple dimensions of model/gradient parameters by exploiting the spatial multiplexing gain.
	In such scenarios, the receive and transmit beamforming should be carefully designed to distinguish spatially multiplexed signals while compensating for channel fading and mitigating the inter-cell interference, which will be studied in our future work.
\end{remark}

\section{Convergence Analysis and Problem Formulation} \label{Sec:CvgPrb}

In this section, we present the convergence analysis of AirComp-assisted FL, taking account of both the downlink and uplink transmission distortions, based on which we formulate a cooperative multi-cell optimization problem to balance the learning performance among different FL tasks in multiple cells.

\subsection{Convergence Analysis}

To proceed, we first make the following two standard assumptions on the loss functions.

\begin{assumption} \label{assump:lower_bound}
	Global loss function $F_m(\cdot)$, $\forall \, m \in \mathcal{M}$, is lower bounded, i.e., $F_m(\bm{w}) \ge F_m(\bm{w}^\star) > -\infty$, $\forall \, \bm{w} \in \mathbb{R}^D$.
\end{assumption}

\begin{assumption} \label{assump:smooth}
	All local loss functions $F_{m, k}(\cdot)$, $\forall \, k \in \mathcal{K}_m$, $\forall \, m \in \mathcal{M}$, are continuously differentiable and their gradients $\nabla F_{m, k}(\cdot)$ are Lipschitz continuous with constant $L > 0$, i.e., for any $\bm{w}, \, \widetilde{\bm{w}} \in \mathbb{R}^D$,
	\begin{align} \label{eq:lipschitz}
		\left\|\nabla F_{m, k}(\bm{w}) - \nabla F_{m, k}(\widetilde{\bm{w}})\right\| \le L \left\|\bm{w} - \widetilde{\bm{w}}\right\|,
	\end{align}
	which is also equivalent to
	\begin{align} \label{eq:smooth}
		& F_{m, k}(\bm{w}) - F_{m, k}(\widetilde{\bm{w}}) \notag \\*
		& \le \left<\nabla F_{m ,k}(\widetilde{\bm{w}}), \bm{w} - \widetilde{\bm{w}}\right> + \frac{L}{2} \left\|\bm{w} - \widetilde{\bm{w}}\right\|^2.
	\end{align}
\end{assumption}

Assumption \ref{assump:lower_bound} is the minimal assumption required for the loss function to converge to a stationary point \cite{allen2018natasha}.
Assumption \ref{assump:smooth} ensures that the local gradients do not change at an arbitrarily high rate with respect to the model parameter \cite{bottou2018optimization}.
This assumption is commonly adopted for the convergence analysis in most existing studies on FL \cite{friedlander2012hybrid, wang2019adaptive, chen2021joint}.
Based on the above assumptions, we present an upper bound of the time-average norm of the global gradients in the following theorem.
\begin{theorem} \label{theorem:convergence}
	In cell $m$, by setting $0 < \eta_m^t \equiv \eta_m < \frac{1}{L}$, the time-average norm of the global gradients after $T$ rounds is upper bounded as
	\begin{align} \label{eq:cvg}
		& \frac{1}{T} \sum_{t = 0}^{T - 1} \mathbb{E}\left[\left\|\nabla F_m(\bm{w}_m^t)\right\|^2\right] \le \underset{\text{Initial gap}}{\underbrace{\frac{2}{\eta_m T} \left(F_m(\bm{w}^0_m) - F_m(\bm{w}_m^\star)\right)}} \notag \\
		& \quad + \underbrace{\frac{1}{T} \sum_{t = 0}^{T - 1} \sum_{d = 1}^{D} \left(\frac{L^2}{K_m} \sum_{k \in \mathcal{K}} \mathbb{E}\left[\Re\left\{(e_{k, d}^{\rm dl})^{t + 1}\right\}^2\right]\right.} \notag \\*
		& \quad \underset{\text{Error-induced gap}}{\underbrace{+ \left. \frac{L \eta_m}{K_m^2} \mathbb{E}\left[\Re\left\{(e_{m, d}^{\rm ul})^{t + 1}\right\}^2\right]\right)}}, \, \forall \, m \in \mathcal{M}
	\end{align}
	where the expectation is taken over normalized transmit symbols and receiver noise.
\end{theorem}
\begin{proof}
	Please refer to Appendix.
\end{proof}

According to Theorem \ref{theorem:convergence}, the upper bound of $\frac{1}{T} \sum_{t = 0}^{T - 1} \mathbb{E}\left[\|\nabla F_m(\bm{w}_m^t)\|^2\right]$ given in \eqref{eq:cvg} consists of two parts, i.e., \textit{initial gap} and \textit{error-induced gap}.
The initial gap is mainly determined by the distance between the values of the global loss function at the initial point and the optimal point, which approaches zero as the number of rounds, $T$, goes to infinity.
Hence, when $T$ is large, the convergence gap is dominated by the error-induced gap, due to the receiver noise, channel fading, and inter-cell interference in both downlink and uplink transmissions.
Inspired by this observation, we aim to minimize the error-induced gap in each time slot for transmitting one-dimension model/gradient parameter in each cell, given by
\begin{align} \label{eq:gap_per_time_slot}
	& \frac{L^2}{K_m} \sum_{k \in \mathcal{K}_m} \mathbb{E}\left[\Re\left\{(e_{k, d}^{\rm dl})^t\right\}^2\right] + \frac{L \eta_m}{K_m^2} \mathbb{E}\left[\Re\left\{(e_{m, d}^{\rm ul})^t\right\}^2\right], \notag \\
	& \hspace{18mm} \forall \, d \in \{1, 2, \dots, D\}, \, \forall \, t \in \mathcal{T}, \, \forall \, m \in \mathcal{M}.
\end{align}
This is due to the fact that, when \eqref{eq:gap_per_time_slot} is minimized, the error-induced gap is minimized.
Nevertheless, simply focusing on minimizing \eqref{eq:gap_per_time_slot} for each cell may lead to severe inter-cell interference in multi-cell wireless networks and in turn deteriorate the learning performance of other cells.
Therefore, a cooperative design is required to balance the learning performance among different FL tasks in multiple cells.

\begin{remark}
	From the second term in \eqref{eq:gap_per_time_slot}, it is observed that by exploiting the gradient aggregation in the uplink, setting a small learning rate $\eta_m$ can reduce the impact of the uplink aggregation error.
	This is because the uplink aggregation error is added on the gradients in \eqref{eq:ul_g}, which enables learning rate $\eta_m$ to rescale the error in \eqref{eq:global_update}.
\end{remark}

\subsection{Problem Formulation}

For presentation clarity, we omit the time indices in the following problem formulation.
%To balance the learning performance among different FL tasks, we first define the gap region, $\mathcal{G}$, to be the set of tuples $(\Delta_1, \Delta_2, \dots, \Delta_M)$, which represents the instantaneous error-induced gaps of all cells that can be simultaneously generated under a given set of downlink and uplink transmit power constraints for BSs and devices, respectively.
To balance the learning performance among different FL tasks, we first define the gap region, $\mathcal{G}$, to be the set of tuples $(\Delta_1, \Delta_2, \dots, \Delta_M)$, which represents the instantaneous error-induced gaps of all cells that can be simultaneously achieved under limited wireless communication resources.
Since the model/gradient parameters are transmitted over the same radio channel in an analog fashion, the communication error mainly stems from the signal misalignment and distortion due to limited transmit power budgets.
Specifically, the gap region, $\mathcal{G}$, can be expressed as
\begin{align}
	\mathcal{G} = \bigcup \left\{\left(\Delta_1, \Delta_2, \dots, \Delta_M\right) \mid \Delta_m \ge \texttt{Gap}_m, \, \forall \, m \in \mathcal{M}\right\}
\end{align}
where
\begin{align} \label{eq:gap_m}
	\texttt{Gap}_m = \underset{\texttt{Gap}_m^{\rm dl}}{\underbrace{\frac{L^2 E_m^{\rm dl}}{K_m}}} + \underset{\texttt{Gap}_m^{\rm ul}}{\underbrace{\frac{L \eta_m E_m^{\rm ul}}{K_m^2}}}
\end{align}
denotes the error-induced gap with
\begin{align} \label{eq:E_dl}
	& E_m^{\rm dl} = \sum_{k \in \mathcal{K}_m} \mathbb{E}\left[\Re\left\{e_{k, d}^{\rm dl}\right\}^2\right] \notag \\
	& = \sum_{k \in \mathcal{K}_m} \nu_m^2 \left(\sum_{l \in \mathcal{M} \setminus \{m\}} \frac{\Re\left\{(h_k^{\rm dl})^\dagger h_{l, k}^{\rm dl}\right\}^2 p_l^{\rm dl}}{|h_k^{\rm dl}|^4 p_m^{\rm dl}} + \frac{(\sigma_k^{\rm dl})^2}{2 |h_k^{\rm dl}|^2 p_m^{\rm dl}}\right)
\end{align}
and
\begin{align} \label{eq:E_ul}
	& E_m^{\rm ul} = \mathbb{E}\left[\Re\left\{e_{m, d}^{\rm ul}\right\}^2\right] = \sum_{k \in \mathcal{K}_m} \left(\frac{|h_k^{\rm ul}| \sqrt{p_k^{\rm ul}}}{\sqrt{c_m}} - \upsilon_k\right)^2 \notag \\
	& \quad + \sum_{l \in \mathcal{M} \setminus \{m\}} \sum_{k' \in \mathcal{K}_l} \frac{\Re\left\{h_{k', m}^{\rm ul} (h_{k'}^{\rm ul})^\dagger\right\}^2 p_{k'}^{\rm ul}}{|h_{k'}^{\rm ul}|^2 c_m} + \frac{(\sigma_m^{\rm ul})^2}{2 c_m}.
\end{align}
Here, \eqref{eq:E_dl} reflects the impact of the sum of downlink dissemination errors at the devices in cell $m$ and \eqref{eq:E_ul} reflects the impact of the uplink aggregation error at BS $m$ on $\texttt{Gap}_m$.
Subsequently, $\texttt{Gap}_m$ is determined by
\begin{align}
	\mathcal{V} = & \left\{\left.\left(\bigcup_{m \in \mathcal{M}} \left\{p_m^{\rm dl}, c_m, \left\{p_k^{\rm ul}\right\}_{k \in \mathcal{K}_m}\right\}\right) \right| 0 \le p_m^{\rm dl} \le P_m^{\rm dl}, \right. \notag \\
	& \left. c_m > 0, 0 \le p_k^{\rm ul} \le P_k^{\rm ul}, \, \forall \, k \in \mathcal{K}_m, \, \forall \, m \in \mathcal{M}\right\}
\end{align}
where $P_m^{\rm dl}$ and $P_k^{\rm ul}$ denote the maximum transmit power budgets in the downlink and uplink, respectively.

As the value of $\texttt{Gap}_m$ in each cell is influenced by the inter-cell interference, there is a performance trade-off among different FL tasks in multiple cells.
In particular, reducing one cell's error-induced gap may increase the error-induced gaps of other cells.
Therefore, we aim to find a feasible set, $\mathcal{V}$, to reach Pareto boundary $\mathcal{P}$ of gap region $\mathcal{G}$ for balancing the learning performance among multiple cells, where the Pareto optimality of a tuple $(\Delta_1, \Delta_2, \dots, \Delta_M)$ is defined as follows \cite{jorswieck2008complete}.

\begin{definition}
	Tuple $(\Delta_1, \Delta_2, \dots, \Delta_M)$ is Pareto optimal if there is no other tuple $(\widetilde{\Delta}_1, \widetilde{\Delta}_2, \dots, \widetilde{\Delta}_M)$ with $(\widetilde{\Delta}_1, \widetilde{\Delta}_2, \dots, \widetilde{\Delta}_M) \preceq (\Delta_1, \Delta_2, \dots, \Delta_M)$ and $(\widetilde{\Delta}_1, \widetilde{\Delta}_2, \dots, \widetilde{\Delta}_M) \ne (\Delta_1, \Delta_2, \dots, \Delta_M)$, where $\preceq$ denotes the component-wise inequality.
\end{definition}

\begin{figure}[t]
\centering
\includegraphics[scale=0.5]{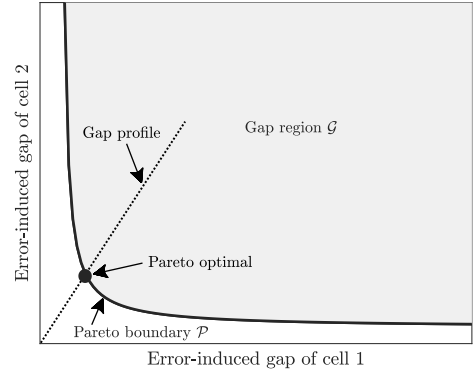}
\caption{Illustration of Pareto boundary $\mathcal{P}$ of gap region $\mathcal{G}$ in a two-cell wireless network.} 
\label{fig:pareto}
\end{figure}

A two-cell example is shown in Fig. \ref{fig:pareto}, where the gray area denotes gap region $\mathcal{G}$ and its lower-left boundary represents Pareto boundary $\mathcal{P}$.
On such a boundary, we can only reduce one cell's error-induced gap at the cost of increasing the error-induced gap of the other cell.
Here, we leverage the profiling technique \cite{zhang2010cooperative, cao2021cooperative} to characterize the Pareto boundary by coordinating all BSs to minimize the sum of error-induced gaps of all cells.
Specifically, let $\bm{\kappa} = \left[\kappa_1, \kappa_2, \dots, \kappa_M\right]$ denote a given profiling vector, which satisfies $\kappa_m \ge 0$, $\forall \, m \in \mathcal{M}$, and $\sum_{m \in \mathcal{M}} \kappa_m = 1$.
The gap tuple on Pareto boundary $\mathcal{P}$ can be obtained by solving the following problem
\begin{subequations} \label{p:origin}
\begin{align}
	\underset{\{p_m^{\rm dl}\}, \{p_k^{\rm ul}\}}{\underset{\zeta, \{c_m\},}{\text{minimize}}} \quad & \zeta \\
	\text{subject to} \quad & \texttt{Gap}_m \le \kappa_m \zeta, \, \forall \, m \in \mathcal{M} \label{cons:gap_m} \\
	& 0 \le p_m^{\rm dl} \le P^{\rm dl}_m, \, \forall \, m \in \mathcal{M} \label{cons:power_dl} \\
	& 0 \le p_k^{\rm ul} \le P^{\rm ul}_k, \, \forall \, k \in \mathcal{K}_m, \, \forall \, m \in \mathcal{M} \label{cons:power_ul} \\
	& c_m > 0, \, \forall \, m \in \mathcal{M} \label{cons:denoising_ul} \\
	& \zeta \ge 0 \label{cons:zeta}
\end{align}
\end{subequations}
where $\zeta$ denotes the sum of error-induced gaps of all cells.
Thus, the gap tuple can be represented as $(\Delta_1, \Delta_2, \dots, \Delta_M) = (\kappa_1 \zeta, \kappa_2 \zeta, \dots, \kappa_M \zeta)$, where a smaller value of $\kappa_m$ implies a more stringent requirement for the error-induced gap of cell $m$.
The corresponding Pareto optimal gap tuple can be geometrically viewed as the intersection of the ray in the direction of $\bm{\kappa}$ (gap profile) and Pareto boundary $\mathcal{P}$, as shown in Fig. \ref{fig:pareto}.

Since the downlink and uplink transmissions are carried out sequentially, the local gradients for uplink aggregation can be obtained only after obtaining the global model from the downlink dissemination, and vice versa. 
As a result, we cannot obtain the exact expressions for constraints \eqref{cons:gap_m}, which makes it highly intractable to simultaneously tackle the downlink and uplink optimization in problem \eqref{p:origin}.
Besides, due to the presence of the inter-cell interference, the downlink and uplink transmit power levels of different cells are coupled in constraints \eqref{cons:gap_m}.
The uplink transmit powers and receive normalizing factor in each cell are also coupled in constraints \eqref{cons:gap_m}.
To address these challenging issues, we propose a cooperative multi-cell FL optimization framework in the following.

\section{Cooperative Multi-Cell FL Optimization Framework} \label{Sec:SysOpt}

In this section, we present a cooperative multi-cell FL optimization framework to balance the learning performance among different FL tasks in multiple cells.

Denote $\zeta = \zeta^{\rm dl}_0 + \zeta^{\rm ul}_0$, where $\zeta^{\rm dl}_0$ and $\zeta^{\rm ul}_0$ are used to quantify the sum of instantaneous error-induced gaps generated by downlink and uplink transmissions, respectively. 
Hence, we rewrite problem \eqref{p:origin} as
\begin{subequations} \label{p:origin_2}
	\begin{align}
		\underset{\{p_m^{\rm dl}\}, \{p_k^{\rm ul}\}}{\underset{\zeta^{\rm dl}_0, \zeta^{\rm ul}_0, \{c_m\},}{\text{minimize}}} \quad & \zeta^{\rm dl}_0 + \zeta^{\rm ul}_0 \\
		\text{subject to} \quad & \texttt{Gap}_m^{\rm dl} \le \kappa_m \zeta^{\rm dl}_0, \, \forall \, m \in \mathcal{M} \label{cons:e_dl} \\
		& \texttt{Gap}_m^{\rm ul} \le \kappa_m \zeta^{\rm ul}_0, \, \forall \, m \in \mathcal{M} \label{cons:e_ul} \\
		& \zeta^{\rm dl}_0 \ge 0 \label{cons:zeta_dl} \\
		& \zeta^{\rm ul}_0 \ge 0 \label{cons:zeta_ul} \\
		& {\rm constraints} \,\, \eqref{cons:power_dl}, \, \eqref{cons:power_ul}, \, \eqref{cons:denoising_ul}.
	\end{align}
\end{subequations}
The downlink and uplink transmissions are decoupled in problem \eqref{p:origin_2}, which allows us to separately optimize the downlink and uplink transmissions.

\subsection{Cooperative Downlink Transmission}

For the downlink model dissemination, the optimization problem can be written as
\begin{subequations} \label{p:dl_origin}
	\begin{align}
		\underset{\zeta^{\rm dl}_0, \{p_m^{\rm dl}\}}{\text{minimize}} \quad & \zeta^{\rm dl}_0 \\
		\text{subject to} \quad & {\rm constraints} \,\, \eqref{cons:power_dl}, \, \eqref{cons:e_dl}, \, \eqref{cons:zeta_dl}.
	\end{align}
\end{subequations}
By substituting \eqref{eq:E_dl} into constraints \eqref{cons:e_dl}, we have
\begin{subequations} \label{p:dl_origin2}
\begin{align}
	\underset{\zeta^{\rm dl}, \{p_m^{\rm dl}\}}{\text{minimize}} \quad & \zeta^{\rm dl} \\
	\text{subject to} \quad & \sum_{l \in \mathcal{M} \setminus \{m\}} p_l^{\rm dl} \sum_{k \in \mathcal{K}_m} \frac{\Re\left\{(h_k^{\rm dl})^\dagger h_{l, k}^{\rm dl}\right\}^2}{|h_k^{\rm dl}|^4} \notag \\
	& \quad + \sum_{k \in \mathcal{K}_m} \frac{(\sigma_k^{\rm dl})^2}{2 |h_k^{\rm dl}|^2} \le \frac{\kappa_m^{\rm dl} \zeta^{\rm dl}}{\nu_m^2} p_m^{\rm dl}, \, \forall \, m \in \mathcal{M} \\
	& 0 \le p_m^{\rm dl} \le P^{\rm dl}_m, \, \forall \, m \in \mathcal{M} \\
	& \zeta^{\rm dl} \ge 0
\end{align}
\end{subequations}
where $\zeta^{\rm dl} = \frac{\zeta_0^{\rm dl}}{L^2}$ and $\kappa_m^{\rm dl} = \kappa_m K_m$.
Problem \eqref{p:dl_origin2} can be solved by tackling a sequence of feasibility detection problems, given $\zeta^{\rm dl}$.
Specifically, for any given $\zeta^{\rm dl}$, problem \eqref{p:dl_origin2} is reduced to
\begin{subequations} \label{p:find_p_dl}
\begin{align}
	\text{find} \quad & \{p_m^{\rm dl}\} \\
	\text{subject to} \quad & \sum_{l \in \mathcal{M} \setminus \{m\}} p_l^{\rm dl} \sum_{k \in \mathcal{K}_m} \frac{\Re\left\{(h_k^{\rm dl})^\dagger h_{l, k}^{\rm dl}\right\}^2}{|h_k^{\rm dl}|^4} \notag \\
	& \quad + \sum_{k \in \mathcal{K}_m} \frac{(\sigma_k^{\rm dl})^2}{2 |h_k^{\rm dl}|^2} \le \frac{\kappa_m^{\rm dl} \zeta^{\rm dl}}{\nu_m^2} p_m^{\rm dl}, \, \forall \, m \in \mathcal{M} \label{cons:find_p_dl} \\
	& \quad 0 \le p_m^{\rm dl} \le P_m^{\rm dl}, \, \forall \, m \in \mathcal{M}.
\end{align}
\end{subequations}
Here, constraints \eqref{cons:find_p_dl} can be rewritten as
\begin{align}
	\sqrt{\varpi_{m, 0}^2 + \sum_{l \in \mathcal{M}} \varpi_{m, l}^2 p_l^{\rm dl}} \le \sqrt{\frac{\kappa_m^{\rm dl} \zeta^{\rm dl}}{\nu_{0, m}^2} p_m^{\rm dl}}, \, \forall \, m \in \mathcal{M} \label{cons:find_p_dl_2}
\end{align}
where
\begin{align}
	\varpi_{m, 0}^2 & = \sum_{k \in \mathcal{K}_m} \frac{(\sigma_k^{\rm dl})^2}{2 |h_k^{\rm dl}|^2}, \, \varpi_{m, m}^2 = 0 \notag \\
	\varpi_{m, l}^2 & = \sum_{k \in \mathcal{K}_m} \frac{\Re\left\{(h_k^{\rm dl})^\dagger h_{l, k}^{\rm dl}\right\}^2}{|h_k^{\rm dl}|^4}, \, \forall \, l \in \mathcal{M} \setminus \{m\}.
\end{align}

By denoting $\bm{q}^{\rm dl} = \left[\sqrt{p_1^{\rm dl}}, \sqrt{p_2^{\rm dl}}, \dots, \sqrt{p_M^{\rm dl}}\right]^{\sf T}$, $\bm{\Pi}_m = {\rm diag}\left(\left[\varpi_{m, 0}, \varpi_{m, 1}, \dots, \varpi_{m, M}\right]\right)$, and $\bm{\vartheta}_m = \left[\vartheta_{m, 1}, \vartheta_{m, 2}, \dots, \vartheta_{m, M}\right]^{\sf T}$ with
\begin{align}
	\vartheta_{m, l} & = \left\{\begin{array}{ll}
		\sqrt{\frac{\kappa_m^{\rm dl} \zeta^{\rm dl}}{\nu_m^2}}, & l = m \\
		0, & l \ne m
	\end{array}\right.
\end{align}
constraints \eqref{cons:find_p_dl_2} can be represented as a set of second-order cone (SOC) constraints.
Problem \eqref{p:find_p_dl} is thus equivalent to the following SOCP problem
\begin{subequations} \label{p:find_q_dl}
\begin{align}
	\text{find} \quad & \bm{q}^{\rm dl} \\*
	\text{subject to} \quad & \left\|[1; \bm{q}^{\rm dl}]^{\sf T} \bm{\Pi}_m\right\| \le (\bm{q}^{\rm dl})^{\sf T} \bm{\vartheta}_m, \, \forall \, m \in \mathcal{M} \\*
	& 0 \le q_m \le \sqrt{P_m^{\rm dl}}, \, \forall \, m \in \mathcal{M}
\end{align}
\end{subequations}
which can be efficiently solved by using convex optimization tools, e.g., CVX \cite{cvx}.
Thus, the optimal downlink transmit power control given $\zeta^{\rm dl}$ can be obtained by $(p_m^{\rm dl})^* = ((q_m^{\rm dl})^*)^2$, $\forall \, m \in \mathcal{M}$, where $(\bm{q}^{\rm dl})^* = [(q_1^{\rm dl})^*, (q_2^{\rm dl})^*, \dots, (q_M^{\rm dl})^*]^{\sf T}$ is the optimal solution to problem \eqref{p:find_q_dl}.
Furthermore, by predefining a solution precision, $\epsilon^{\rm dl} > 0$, we can leverage the bisection search method to obtain optimal solutions $(\zeta^{\rm dl})^\star$ and $\{(p_m^{\rm dl})^\star\}_{m \in \mathcal{M}}$ via checking the feasibility of the solution to problem \eqref{p:find_q_dl}, which is summarized in Algorithm \ref{alg:downlink}.

\begin{algorithm}[t]
	\caption{Cooperative downlink optimization for problem \eqref{p:dl_origin}.}
	\label{alg:downlink}
	
	\KwIn{Standard deviations of global models $\{\nu_m\}$, profiling vector $\bm{\kappa}$, and solution precision parameter $\epsilon^{\rm dl}$.}
	
	$\zeta^{\rm dl}_{\rm low} \leftarrow 0$, $\zeta^{\rm dl}_{\rm up} \leftarrow 1$. \\
	\While{$\left|\zeta^{\rm dl}_{\rm up} - \zeta^{\rm dl}_{\rm low}\right| > \epsilon^{\rm dl}$}{
		$\zeta^{\rm dl} \leftarrow \frac{\zeta^{\rm dl}_{\rm up} + \zeta^{\rm dl}_{\rm low}}{2}$. \\
		Obtain $\{(p_m^{\rm dl})^*\}_{m \in \mathcal{M}}$ by solving problem \eqref{p:find_q_dl}. \\
		\uIf{Problem \eqref{p:find_q_dl} is feasible}{
			$\zeta^{\rm dl}_{\rm up} \leftarrow \zeta^{\rm dl}$. \\
			$(p_m^{\rm dl})^\star \leftarrow (p_m^{\rm dl})^*, \, \forall \, m \in \mathcal{M}$.
		}
		\Else{
			$\zeta^{\rm dl}_{\rm low} \leftarrow \zeta^{\rm dl}$.
		}
	}
	
	\KwOut{$(\zeta^{\rm dl})^\star \leftarrow \frac{\zeta^{\rm dl}_{\rm up} + \zeta^{\rm dl}_{\rm low}}{2}$ and optimal downlink transmit powers $\{(p_m^{\rm dl})^\star\}_{m \in \mathcal{M}}$.}
\end{algorithm}

\subsection{Cooperative Uplink Transmission}

For the uplink gradient aggregation, the optimization problem is given by
\begin{subequations} \label{p:ul_origin}
	\begin{align}
		\underset{\zeta^{\rm ul}_0, \{c_m\}, \{p_k^{\rm ul}\}}{\text{minimize}} \quad & \zeta^{\rm ul}_0 \\*
		\text{subject to} \quad & {\rm constraints} \,\, \eqref{cons:power_ul}, \, \eqref{cons:denoising_ul}, \, \eqref{cons:e_ul}, \, \eqref{cons:zeta_ul}.
	\end{align}
\end{subequations}
It is intractable to directly solve problem \eqref{p:ul_origin} due to the coupling between uplink transmit power levels $\bigcup_{m \in \mathcal{M}} \{p_k^{\rm ul}\}_{k \in \mathcal{K}_m}$ and receive normalizing factors $\{c_m\}_{m \in \mathcal{M}}$.
To tackle this issue, we first obtain the optimal receive normalizing factors with fixed uplink transmit power levels.
Specifically, given $\bigcup_{m \in \mathcal{M}} \{p_k^{\rm ul}\}_{k \in \mathcal{K}_m}$, problem \eqref{p:ul_origin} can be decoupled into $M$ subproblems, where the subproblem for cell $m$ is given by
\begin{align} \label{p:c_m}
	\underset{c_m > 0}{\text{minimize}} \quad & \sum_{k \in \mathcal{K}_m} \left(\frac{|h_k^{\rm ul}| \sqrt{p_k^{\rm ul}}}{\sqrt{c_m}} - \upsilon_k\right)^2 + \frac{(\sigma_m^{\rm ul})^2}{2 c_m} \notag \\*
	& + \sum_{l \in \mathcal{M} \setminus \{m\}} \sum_{k' \in \mathcal{K}_l} \frac{\Re\left\{h_{k', m}^{\rm ul} (h_{k'}^{\rm ul})^\dagger\right\}^2 p_{k'}^{\rm ul}}{|h_{k'}^{\rm ul}|^2 c_m}.
\end{align}
It is easy to verify that problem \eqref{p:c_m} is a convex quadratic problem by treating $\frac{1}{\sqrt{c_m}}$ as an optimization variable.
Hence, the optimal receive normalizing factors can be computed as
\begin{align} \label{eq:c_m}
	c_m^\star & = \left(\frac{\sum_{k \in \mathcal{K}_m} |h_k^{\rm ul}|^2 p_k^{\rm ul} + (\sigma_m^{\rm ul})^2 / 2}{\sum_{k \in \mathcal{K}_m} |h_k^{\rm ul}| \sqrt{p_k^{\rm ul}} \upsilon_k} \right. \notag \\*
	& \left. + \frac{\sum_{l \in \mathcal{M} \setminus \{m\}} \sum_{k' \in \mathcal{K}_l} \Re\left\{h_{k', m}^{\rm ul} (h_{k'}^{\rm ul})^\dagger\right\}^2 p_{k'}^{\rm ul} / |h_{k'}^{\rm ul}|^2}{\sum_{k \in \mathcal{K}_m} |h_k^{\rm ul}| \sqrt{p_k^{\rm ul}} \upsilon_k}\right)^2, \notag \\*
	& \hspace{59mm} \forall \, m \in \mathcal{M}.
\end{align}

On the other hand, substituting \eqref{eq:c_m} into problem \eqref{p:ul_origin}, we have
\begin{subequations} \label{p:ul_origin2}
\begin{align}
	\underset{\zeta^{\rm ul}, \{p_k^{\rm ul}\}}{\text{minimize}} \quad & \zeta^{\rm ul} \\
	\text{subject to} \quad & \left(\sum_{k \in \mathcal{K}_m} \upsilon_k^2 - \kappa_m^{\rm ul} \zeta^{\rm ul}\right) \left(\sum_{k \in \mathcal{K}_m} |h_k^{\rm ul}|^2 p_k^{\rm ul} + \frac{(\sigma_m^{\rm ul})^2}{2} \right. \notag \\
	& \quad \left. + \sum_{l \in \mathcal{M} \setminus \{m\}} \sum_{k' \in \mathcal{K}_l} \frac{\Re\left\{h_{k', m}^{\rm ul} (h_{k'}^{\rm ul})^\dagger\right\}^2 p_{k'}^{\rm ul}}{|h_{k'}^{\rm ul}|^2}\right) \notag \\
	& \le \left(\sum_{k \in \mathcal{K}_m} |h_k^{\rm ul}| \sqrt{p_k^{\rm ul}} \upsilon_k\right)^2, \, \forall \, m \in \mathcal{M} \\
	& 0 \le p_k^{\rm ul} \le P_k^{\rm ul}, \, \forall \, k \in \mathcal{K}_m, \, \forall \, m \in \mathcal{M} \\
	& \zeta^{\rm ul} \ge 0
\end{align}
\end{subequations}
where $\zeta^{\rm ul} = \frac{\zeta_0^{\rm ul}}{L}$ and $\kappa_m^{\rm ul} = \frac{\kappa_m K_m^2}{\eta_m}$.
This problem can be solved in a way similar to solving problem \eqref{p:dl_origin2}.
Specifically, given $\zeta^{\rm ul}$, problem \eqref{p:ul_origin2} becomes
\begin{subequations} \label{p:find_p_ul}
\begin{align}
	\text{find} \quad & \{p_k^{\rm ul}\} \\*
	\text{subject to} \quad & \Xi_m \left(\sum_{k \in \mathcal{K}_m} |h_k^{\rm ul}|^2 p_k^{\rm ul} + \frac{(\sigma_m^{\rm ul})^2}{2} \right. \notag \\*
	& \quad \left. + \sum_{l \in \mathcal{M} \setminus \{m\}} \sum_{k' \in \mathcal{K}_l} \frac{\Re\left\{h_{k', m}^{\rm ul} (h_{k'}^{\rm ul})^\dagger\right\}^2 p_{k'}^{\rm ul}}{|h_{k'}^{\rm ul}|^2}\right) \notag \\*
	& \le \left(\sum_{k \in \mathcal{K}_m} |h_k^{\rm ul}| \sqrt{p_k^{\rm ul}} \upsilon_k\right)^2, \, \forall \, m \in \mathcal{M} \label{cons:find_p_ul} \\*
	& 0 \le p_k^{\rm ul} \le P_k^{\rm ul}, \, \forall \, k \in \mathcal{K}_m, \, \forall \, m \in \mathcal{M}
\end{align}
\end{subequations}
where $\Xi_m = \sum_{k \in \mathcal{K}_m} \upsilon_k^2 - \kappa_m^{\rm ul} \zeta^{\rm ul} \ge 0$.
By defining $\bm{\gamma}_m = [\gamma_{m, 1}, \gamma_{m, 2}, \dots, \gamma_{m, K_{\rm tot}}]$, $\forall \, m \in \mathcal{M}$, with $K_{\rm tot} = \sum_{m \in \mathcal{M}} K_m$ and
\begin{align}
	\gamma_{m, k} = \left\{\begin{array}{ll}
		|h_k^{\rm ul}|, &  k \in \mathcal{K}_m \\
		\frac{\Re\left\{h_{k, m}^{\rm ul} (h_{k}^{\rm ul})^\dagger\right\}}{|h_{k}^{\rm ul}|}, & k \in \mathcal{K}_l, \, \forall \, l \in \mathcal{M} \setminus \{m\}
	\end{array}\right.
\end{align}
constraints \eqref{cons:find_p_ul} can be reformulated as
\begin{align}
	& \sqrt{\Xi_m \left(\sum_{l \in \mathcal{M}} \sum_{k \in \mathcal{K}_l} \gamma_{m, k}^2 p_k^{\rm ul} + \frac{(\sigma_m^{\rm ul})^2}{2}\right)} \le \sum_{k \in \mathcal{K}_m} |h_k^{\rm ul}| \sqrt{p_k^{\rm ul}} \upsilon_k, \notag \\
	& \hspace{64mm} \forall \, m \in \mathcal{M}.
\end{align}

Then, problem \eqref{p:find_p_ul} can be rewritten as the following SOCP problem:
\begin{subequations} \label{p:find_q_ul}
\begin{align}
	\text{find} \quad & \bm{q}^{\rm ul} \\*
	\text{subject to} \quad & \sqrt{\Xi_m} \left\|[1; \bm{q}^{\rm ul}]^{\sf T} \bm{\Gamma}_m\right\| \le (\bm{q}^{\rm ul})^{\sf T} \bm{\psi}_m, \, \forall \, m \in \mathcal{M} \\*
	& 0 \le q_k^{\rm ul} \le \sqrt{P_k^{\rm ul}}, \, \forall \, k \in \mathcal{K}_m, \, \forall \, m \in \mathcal{M}
\end{align}
\end{subequations}
where $\bm{q}^{\rm ul} = \left[\sqrt{p_1^{\rm ul}}, \sqrt{p_2^{\rm ul}}, \dots, \sqrt{p_{K_{\rm tot}}^{\rm ul}}\right]^{\sf T}$, $\bm{\Gamma}_m = {\rm diag}\left(\left[\sigma_m^{\rm ul} / \sqrt{2}, \bm{\gamma}_m\right]\right)$, and $\bm{\psi}_m = \left[\psi_{m, 1}, \psi_{m, 2}, \dots, \psi_{m, K_{\rm tot}}\right]^{\sf T}$ with
\begin{align}
	\psi_{m, k} = \left\{\begin{array}{ll}
		|h_k^{\rm ul}| \upsilon_k, & k \in \mathcal{K}_m \\
		0, & k \in \mathcal{K}_l, \, \forall \, l \in \mathcal{M} \setminus \{m\}.
	\end{array}\right.
\end{align}
Consequently, we can obtain the optimal uplink transmit power as $(p_k^{\rm ul})^* = ((q_k^{\rm ul})^*)^2$, $\forall \, k \in \mathcal{K}_m$, $\forall \, m \in \mathcal{M}$, given $\zeta^{\rm ul}$, by solving problem \eqref{p:find_q_ul} with convex optimization tools, where $(\bm{q}^{\rm ul})^* = [(q_1^{\rm ul})^*, (q_2^{\rm ul})^*, \dots, (q_{K_{\rm tot}}^{\rm ul})^*]^{\sf T}$ is the optimal solution to problem \eqref{p:find_q_ul}.
Moreover, by given solution precision $\epsilon^{\rm ul}$, the bisection search method can be used to find the optimal solutions $(\zeta^{\rm ul})^\star$ and $\bigcup_{m \in \mathcal{M}} \{(p_k^{\rm ul})^\star\}_{k \in \mathcal{K}_m}$ to problem \eqref{p:ul_origin2}, while optimal receive normalizing factors $\{c_m^\star\}_{m \in \mathcal{M}}$ can be obtained by substituting $\bigcup_{m \in \mathcal{M}} \{(p_k^{\rm ul})^\star\}_{k \in \mathcal{K}_m}$ into \eqref{eq:c_m}.
The uplink transmission optimization for problem \eqref{p:ul_origin} is summarized in Algorithm \ref{alg:uplink}.

\begin{remark}
	From \eqref{eq:c_m}, it is observed that optimal receive normalizing factors $\{c_m^\star\}$ are inversely proportional to standard deviations $\{\upsilon_k\}$ of local gradient parameters.
	Since the absolute value of each entry of the gradient parameter tends to become smaller when the model parameter approaches the stationary point of the loss function, the standard deviation of all entries of the gradient parameter is also expected to become smaller.
	Besides, \eqref{eq:r_ul_m_d} demonstrates that a larger $\{c_m^\star\}$ can better suppress the uplink aggregation error caused by the receiver noise and inter-cell interference.
	Therefore, transmitting local gradients in the uplink is able to achieve better resistance to the receiver noise and inter-cell interference in the later rounds of training process under the same system constraints.
\end{remark}

\begin{algorithm}[t]
	\caption{Cooperative uplink optimization for problem \eqref{p:ul_origin}.}
	\label{alg:uplink}
	
	\KwIn{Standard deviations of local gradients $\{\upsilon_k\}$, profiling vector $\bm{\kappa}$, and solution precision parameter $\epsilon^{\rm ul}$.}
	
	$\zeta^{\rm ul}_{\rm low} \leftarrow 0$, $\zeta^{\rm ul}_{\rm up} \leftarrow \min_{m \in \mathcal{M}} \frac{\sum_{k \in \mathcal{K}_m} \upsilon_k^2}{\kappa_m^{\rm ul}}$. \\
	\While{$\left|\zeta^{\rm ul}_{\rm up} - \zeta^{\rm ul}_{\rm low}\right| > \epsilon^{\rm ul}$}{
		$\zeta^{\rm ul} \leftarrow \frac{\zeta^{\rm ul}_{\rm up} + \zeta^{\rm ul}_{\rm low}}{2}$. \\
		Obtain $\bigcup_{m \in \mathcal{M}} \{(p_m^{\rm ul})^*\}_{k \in \mathcal{K}_m}$ by solving problem \eqref{p:find_q_dl}. \\
		\uIf{Problem \eqref{p:find_q_dl} is feasible}{
			$\zeta^{\rm ul}_{\rm up} \leftarrow \zeta^{\rm ul}$. \\
			$(p_k^{\rm ul})^\star \leftarrow (p_k^{\rm ul})^*, \, \forall \, k \in \mathcal{K}_m, \, \forall \, m \in \mathcal{M}$.
		}
		\Else{
			$\zeta^{\rm ul}_{\rm low} \leftarrow \zeta^{\rm ul}$.
		}
	}
	Obtain $\{c_m^\star\}$ based on \eqref{eq:c_m}.
	
	\KwOut{$(\zeta^{\rm ul})^\star \leftarrow \frac{\zeta^{\rm ul}_{\rm up} + \zeta^{\rm ul}_{\rm low}}{2}$, optimal uplink transmit powers $\bigcup_{m \in \mathcal{M}} \{(p_k^{\rm ul})^\star\}_{k \in \mathcal{K}_m}$, and optimal receive normalizing factors $\{c_m^\star\}_{m \in \mathcal{M}}$.}
\end{algorithm}

\subsection{Computational Complexity Analysis}

According to \cite{polik2010interior, beko2012efficient}, the computational complexity of solving an SOCP problem with the interior-point method is $\mathcal{O}\left(k_{\rm soc}^{0.5} \left(m_{\rm soc}^3 + m_{\rm soc}^2 \sum_{j = 1}^{k_{\rm soc}} n_{{\rm soc}, j} + \sum_{j = 1}^{k_{\rm soc}} n_{{\rm soc}, j}^2\right)\right)$, where $k_{\rm soc}$ denotes the number of SOC constraints, $m_{\rm soc}$ the number of equality constraints, and $n_{{\rm soc}, j}$ the dimension of the $j$-th SOC.
In Algorithm \ref{alg:downlink}, we need to solve a series of SOCP problems of \eqref{p:find_q_dl}, where each SOCP problem has $k_{\rm soc} = M$, $m_{\rm soc} = 0$, and $n_{{\rm soc}, 1:k_{\rm soc}} = M + 2$, thereby requiring a computational complexity of $\mathcal{O}\left(M^{3.5}\right)$.
In addition, as the bisection search method generally takes $\mathcal{O}\left(\rm{log}\left(\left(\zeta^{\rm dl}_{\rm up} - \zeta^{\rm dl}_{\rm low}\right) / \epsilon^{\rm dl}\right)\right)$ iterations, the computational complexity of the downlink optimization is $\mathcal{O}\left({\rm log}\left(\left(\zeta^{\rm dl}_{\rm up} - \zeta^{\rm dl}_{\rm low}\right) / \epsilon^{\rm dl}\right) M^{3.5}\right)$.
Similarly, the computational complexity of Algorithm \ref{alg:uplink} is $\mathcal{O}\left({\rm log}\left(\left(\zeta^{\rm ul}_{\rm up} - \zeta^{\rm ul}_{\rm low}\right) / \epsilon^{\rm ul}\right) M^{1.5} K_{\rm tot}^2\right)$.

\subsection{Discussion on Implementation Issues}

\begin{figure*}[t]
	\centering
	\includegraphics[scale=0.55]{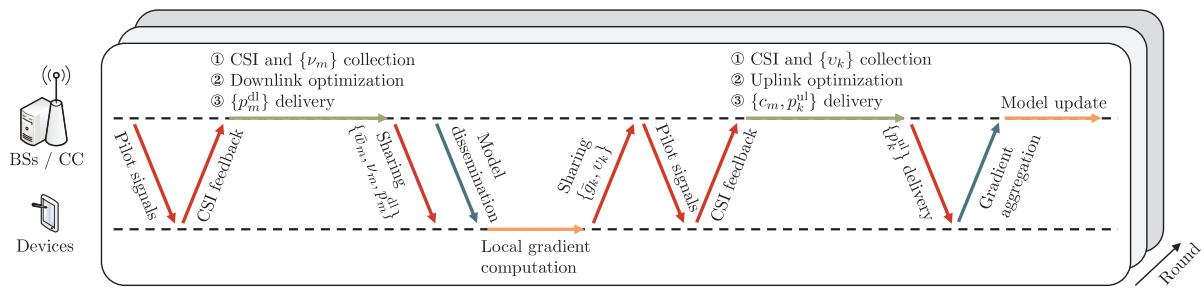}
	\caption{Communication protocol within a typical training round of our proposed over-the-air FL in a multi-cell wireless network.} 
	\label{fig:protocol}
\end{figure*}

To implement the proposed cooperative multi-cell FL optimization framework, the communication protocol illustrated in Fig. \ref{fig:protocol} can be adopted for each training round.

Before the downlink model dissemination, each BS first broadcasts orthogonal pilot signals to the devices.
Each device estimates and feeds back the channel conditions between itself and the BSs to its home BS, e.g., device $k \in \mathcal{K}_m$ feeds back $h_k^{\rm dl}$ and $\{h_{l, k}^{\rm dl}\}_{l \in \mathcal{M} \setminus \{m\}}$ to BS $m$.
Since the BSs are generally connected through high-rate wired backhaul links, one of the BSs can be designated as a centralized controller (CC) \cite{zarikoff2010coordinated} for collecting total downlink CSI $\bigcup_{m \in \mathcal{M}} \{h_k^{\rm dl}, \{h_{l, k}^{\rm dl}\}_{l \in \mathcal{M} \setminus \{m\}}\}_{k \in \mathcal{K}_m}$ and standard deviations $\{\nu_m\}_{m \in \mathcal{M}}$ of the model parameters from other BSs for the cooperative downlink optimization.
After that, the CC delivers the setting of downlink transmit power $\{p_m^{\rm dl}\}_{m \in \mathcal{M}}$ to the corresponding BSs, and then BS $m$ shares model statistics $\{\bar{w}_m, \nu_m\}$ and power setting $p_m^{\rm dl}$ to its associated devices for model recovering as in \eqref{eq:r_dl_k_d}.
Consequently, each device receives a noisy version of the model parameter disseminated by its home BS, and accordingly computes the local gradient for the next model update.

After finishing the local gradient computation, the devices located in cell $m$ first upload their gradient statistics, $\{\bar{g}_k, \upsilon_k\}_{k \in \mathcal{K}_m}$, to BS $m$, which will be utilized for the cooperative uplink optimization and the gradient aggregation as shown in \eqref{eq:r_ul_m_d}.
Then, the CC leverages the same method as in the downlink transmission to collect the total uplink CSI, $\bigcup_{m \in \mathcal{M}} \{h_k^{\rm ul}, \{h_{k, l}^{\rm ul}\}_{l \in \mathcal{M} \setminus \{m\}}\}_{k \in \mathcal{K}_m}$, by exploiting the channel reciprocity, and to gather standard deviations $\bigcup_{m \in \mathcal{M}} \{\upsilon_k\}_{k \in \mathcal{K}_m}$ of the local gradients, based on which the cooperative uplink optimization can be completed.
Subsequently, optimized receive normalizing factors $\{c_m\}_{m \in \mathcal{M}}$ and uplink transmit power $\bigcup_{m \in \mathcal{M}} \{p_k^{\rm ul}\}_{k \in \mathcal{K}_m}$ can be assigned to each BS and its associated devices, respectively.
Finally, each BS aggregates the local gradients of its associated devices as in \eqref{eq:r_ul_m_d}, and performs the global model update as in \eqref{eq:global_update}.
Note that, if the CSI remains invariant within each training round, the overhead for uplink channel estimation can be eliminated.

Since the model/gradient parameters that need to be transmitted are typically high-dimensional during the training process, the above exchanged scalars are relatively small in terms of the packet size and hence are assumed to be transmitted in an error-free manner with negligible overhead \cite{liu2021reconfigurable}.

\begin{remark}
	In addition to the centralized optimization method assisted with a CC, interference management can be achieved by the distributed optimization method, which leverages the interference temperature (IT) technique to limit the interference from other cells \cite{zhang2010cooperative, cao2021cooperative}.
	Specifically, the IT technique allows the inter-cell interference term to be replaced by a constant, namely IT level, which enables each cell to optimize its transmission regardless of the transmission of other cells.
	The IT constraints should be considered in the optimization to ensure that the inter-cell interference does not exceed the IT level.
	Meanwhile, each BS merely needs to obtain the CSI of its home cell, which reduces the overhead for collecting the CSI of other cells.
	The distributed optimization algorithm is left for our future work.
\end{remark}

\section{Simulation Results} \label{Sec:Sim}

In this section, we present the numerical results to evaluate the performance of our proposed cooperative multi-cell FL optimization framework, based on computer simulations.

\subsection{Simulation Setup}

We consider a four-cell wireless network, where the BSs are located at $\left(0, 0\right)$, $\left(40, 0\right)$, $\left(20, 20 \sqrt{3}\right)$, and $\left(20, - 20 \sqrt{3}\right)$ meters, respectively.
The devices in each cell are uniformly and randomly distributed in a circle of radius $\varrho \in \left[1, 20\right]$ meters centered at their home BS.
The number of devices located in each cell is set to $K_m = \bar{K} = 10$, $\forall \, m \in \mathcal{M}$.
All channel coefficients are modeled as \cite{tse2005fundamentals}
\begin{align}
	h = \rho^{- \alpha / 2} \left(\sqrt{\frac{\beta}{1 + \beta}} h_{\rm LoS} + \sqrt{\frac{1}{1 + \beta}} h_{\rm NLoS}\right)
\end{align}
and vary independently over different rounds, where $\rho$ denotes the distance between the transmitter and the receiver, $\alpha = 2.5$ denotes the pathloss exponent, $\beta = 5$ (dB) represents the Rician factor, $h_{\rm LoS}$ denotes the line-of-sight (LoS) component, and $h_{\rm NLoS} \sim \mathcal{CN}(0, 1)$ denotes the non-line-of-sight (NLoS) exponent.
Considering heterogeneity of different BSs and devices, the maximum downlink transmit power budgets of four BSs are set to $P_1^{\rm dl} = 40$ dBm, $P_2^{\rm dl} = 30$ dBm, $P_3^{\rm dl} = 30$ dBm, and $P_4^{\rm dl} = 40$ dBm, respectively.
The maximum transmit power budgets of devices in each cell are set to $P_k^{\rm ul} = 15$ dBm, $\forall \, k \in \{\sum_{l = 1}^{m - 1} K_l + 1, \dots, \sum_{l = 1}^{m - 1} K_l + \lfloor\frac{\bar{K}}{2}\rfloor\}$, and $P_k^{\rm ul} = 30$ dBm, $\forall \, k \in \{\sum_{l = 1}^{m - 1} K_l + \lfloor\frac{\bar{K}}{2}\rfloor + 1, \dots, \sum_{l = 1}^{m - 1} K_l + \bar{K}\}$, $\forall \, m \in \mathcal{M}$,
where $\lfloor x \rfloor$ is the floor function that returns the greatest integer less than or equal to $x \in \mathbb{R}$.
In addition, the noise power are set to $(\sigma_m^{\rm ul})^2 = (\sigma_k^{\rm dl})^2 = -110$ dBm, $\forall \, k \in \mathcal{K}_m$, $\forall \, m \in \mathcal{M}$, and the solution precision is set to $\epsilon^{\rm dl} = \epsilon^{\rm ul} = 10^{-9}$.
All simulation results in the following are obtained by averaging over $100$ experiments.

\subsubsection{Learning Model Setting}

We leverage the multinomial logistic regression to train the learning models with the following specific settings.

\begin{itemize}
\item \textbf{Sample-Wise Loss Function}:
Suppose the labeled data sample can be represented as $\bm{\xi} = [\bm{u}; v]$, where $\bm{u}$ denotes the data feature and $v$ denotes the ground-truth label of $\bm{u}$.
The sample-wise loss function used for the training is given by \cite{jurafsky2009speech}
\begin{align}
	f(\bm{w}; \bm{\xi}) = - \sum_{c = 1}^C \mathbb{I}_{\{v = c\}} {\rm log} \left(\frac{\exp\left({\bm{w}_c^{\sf T} \bm{u}}\right)}{\sum_{j = 1}^C \exp\left({\bm{w}_j^{\sf T} \bm{u}}\right)}\right)
\end{align}
where $C$ denotes the total number of label categories, the model parameter $\bm{w}$ consists of the parameter for each label category, i.e., $\bm{w} = [\bm{w}_1^{\sf T}, \bm{w}_2^{\sf T}, \dots, \bm{w}_C^{\sf T}]^{\sf T}$, and $\mathbb{I}_{\{v = c\}}$ is an indicator function defined as
\begin{align}
	\mathbb{I}_{\{v = c\}} = \left\{\begin{array}{ll}
		1, & v = c \\
		0, & v \ne c.
	\end{array}\right.
\end{align}
The partial gradient with respect to $\bm{w}_c$ is
\begin{align}
	\nabla_{\bm{w}_c} f(\bm{w}; \bm{\xi}) = - \left(\mathbb{I}_{\{v = c\}} - \frac{\exp(\bm{w}_c^{\sf T} \bm{u})}{\sum_{j = 1}^C \exp(\bm{w}_j^{\sf T} \bm{u})}\right) \bm{u}
\end{align}
and the entire gradient can be expressed as $\nabla f(\bm{w}; \bm{\xi}) = [\nabla_{\bm{w}_1} f(\bm{w}; \bm{\xi})^{\sf T}, \nabla_{\bm{w}_2} f(\bm{w}; \bm{\xi})^{\sf T}, \dots, \nabla_{\bm{w}_C} f(\bm{w}; \bm{\xi})^{\sf T}]^{\sf T}$.

\item \textbf{Dataset}:
The MNIST \cite{lecun1998gradient} and Fashion-MNIST \cite{xiao2017fashion} datasets are used in the considered multi-cell FL system, where different cells perform different learning tasks based on their assigned datasets.
Specifically, the data labeled with $0 \sim 4$ in the MNIST dataset are assigned to cell $1$, the data labeled with $5 \sim 9$ in the MNIST dataset are assigned to cell $2$, the data labeled with $0 \sim 4$ in the Fashion-MNIST dataset are assigned to cell $3$, and the data labeled with $5 \sim 9$ in the Fashion-MNIST dataset are assigned to cell $4$.
In each cell, we first sort the dataset by the contained labels, then divide it into $\bar{K}$ shards, and finally assign one shard for each device without replacement.

\item \textbf{Learning Rate}:
The learning rates of four cells are set to $\eta_1 = \eta_2 = 0.1$ and $\eta_3 = \eta_4 = 0.01$.

\end{itemize}

\subsubsection{Baseline Schemes}

We consider the following transmission schemes for comparison.

\begin{itemize}
	\item \textbf{Benchmark / DL-Free / UL-Free}:
	The case that both the downlink model dissemination and uplink gradient aggregation are realized in an error-free manner serves as the \textit{Benchmark} scheme.
	This scheme achieves the best learning performance by assuming the accurate model/gradient exchange.
	The error-free downlink and uplink transmissions are referred to as \textit{DL-Free} and \textit{UL-Free}, respectively.
	
	\item \textbf{DL-Opt / UL-Opt}:
	\textit{DL-Opt} denotes the strategy that performs the resource allocation for the downlink transmission using Algorithm \ref{alg:downlink}, and \textit{UL-Opt} denotes the strategy that performs the resource allocation for the uplink transmission using Algorithm \ref{alg:uplink}.
	These schemes are used to evaluate the learning performance of our proposed cooperative multi-cell FL optimization framework.
	
	\item \textbf{UL-IgnInter}:
	In this case, each cell independently optimizes its uplink transmission by ignoring the inter-cell interference.
	Hence, the optimization problem for cell $m$ is given by
	\begin{subequations} \label{p:ul_igninter}
		\begin{align}
			\underset{c_m, \{p_k^{\rm ul}\}}{\text{minimize}} \quad & \sum_{k \in \mathcal{K}_m} \left(\frac{|h_k^{\rm ul}| \sqrt{p_k^{\rm ul}}}{\sqrt{c_m}} - \upsilon_k\right)^2 + \frac{(\sigma_m^{\rm ul})^2}{2 c_m} \\*
			\text{subject to} \quad & 0 \le p_k^{\rm ul} \le P_k^{\rm ul}, \, \forall \, k \in \mathcal{K}_m, \,  c_m > 0
		\end{align}
	\end{subequations}
	which can be solved by using the method proposed in \cite{cao2020optimized}.
	This scenario reduces to the uplink transmission design for single-cell wireless networks.
	
	\item \textbf{UL-MaxInter}:
	In this case, each cell optimizes its uplink transmission by assuming the existence of the maximum inter-cell interference from other cells.
	Hence, the optimization problem for cell $m$ is given by
	\begin{subequations}
		\begin{align}
			\underset{c_m, \{p_k^{\rm ul}\}}{\text{minimize}} \quad & \sum_{k \in \mathcal{K}_m} \left(\frac{|h_k^{\rm ul}| \sqrt{p_k^{\rm ul}}}{\sqrt{c_m}} - \upsilon_k\right)^2 + \frac{(\sigma_m^{\rm ul})^2}{2 c_m} \notag \\*
			& + \sum_{l \in \mathcal{M} \setminus \{m\}} \sum_{k' \in \mathcal{K}_l} \frac{\Re\left\{h_{k', m}^{\rm ul} (h_{k'}^{\rm ul})^\dagger\right\}^2 P_{k'}^{\rm ul}}{|h_{k'}^{\rm ul}|^2 c_m} \\*
			\text{subject to} \quad & 0 \le p_k^{\rm ul} \le P_k^{\rm ul}, \, \forall \, k \in \mathcal{K}_m, \, c_m > 0
		\end{align}
	\end{subequations}
	which can be solved by using the method proposed in \cite{cao2020optimized}.
	This situation is similar to \textit{UL-IgnInter} that focuses on single-cell transmission design but considering the worst-case inter-cell interference.
	
	\item \textbf{DL-Full / UL-Full}:
	The full-power transmission is applied to the system by setting $p_m^{\rm dl} = P_m^{\rm dl}$, $\forall \, m \in \mathcal{M}$, in the \textit{DL-Full} scheme, and setting $p_k^{\rm ul} = P_k^{\rm ul}$, $\forall \, k \in \mathcal{K}_m$, $\forall \, m \in \mathcal{M}$, in the \textit{UL-Full} scheme, respectively.
\end{itemize}

\subsection{Performance Evaluation}

\subsubsection{Two-Cell Network}

We first evaluate the learning performance of our proposed cooperative multi-cell FL optimization framework in a two-cell network, i.e., only cell $1$ and cell $2$ are active, with profiling vector $\bm{\kappa} = \left[\frac{1}{2}, \frac{1}{2}\right]$.

\begin{figure}[t]
	\centering
	\subfigure[Training loss versus the number of rounds.]{
		\label{fig:dl_loss}
		\centering
		\includegraphics[scale=0.5]{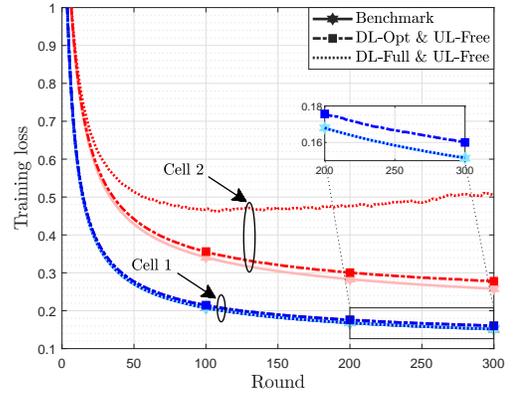}
	}
	\hspace{3mm}
	\subfigure[Test accuracy versus the number of rounds.]{
		\label{fig:dl_acc}
		\centering
		\includegraphics[scale=0.5]{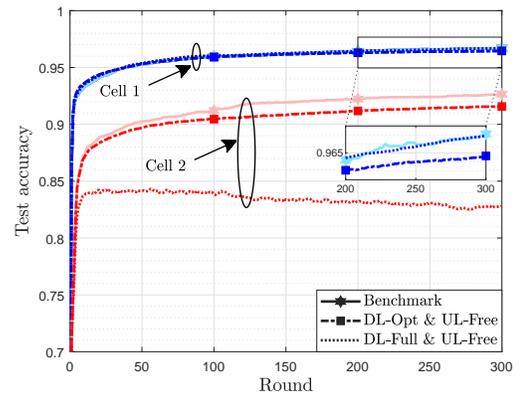}
	}
	\caption{Learning performance comparison for different downlink transmission schemes when the uplink gradient aggregation is error-free.} 
\end{figure}

Figs. \ref{fig:dl_loss} and \ref{fig:dl_acc} show the learning performance versus the number of rounds for different downlink transmission schemes when the uplink gradient aggregation is error-free.
It is observed that the training loss and test accuracy obtained by \textit{DL-Opt} are close to those obtained by \textit{Benchmark} in both two cells, which reveals that our proposed cooperative downlink transmission is able to properly balance the downlink dissemination errors between two cells while ensuring learning performance.
Meanwhile, \textit{DL-Full} yields almost the same training loss and test accuracy as \textit{Benchmark} in cell $1$, but has larger performance gaps in comparison with \textit{Benchmark} in cell $2$.
This indicates that the inter-cell interference from cell $1$ severely deteriorates the downlink transmission in cell $2$ due to the lack of interference management.
The magnitude of downlink dissemination errors at each device increases with the transmission power of its non-associated BSs and decreases with that of its home BS according to \eqref{eq:r_dl_k_d}.
As the maximum downlink transmit power budget of cell $1$ is larger than that of cell $2$, simply implementing the downlink full-power transmission leads to poor inter-cell interference suppression for devices in cell $2$, thereby reducing the accuracy of the received global model at each device.
In addition, the training loss of cell $2$ obtained by \textit{DL-Full} shows a decreasing trend in early rounds but increases later, while the test accuracy first increases but subsequently decreases.
This demonstrates that inaccurate downlink model dissemination leads to worse performance as the number of rounds increases.

\begin{figure}[t]
	\centering
	\subfigure[Training loss versus the number of rounds.]{
		\label{fig:ul_loss}
		\centering
		\includegraphics[scale=0.5]{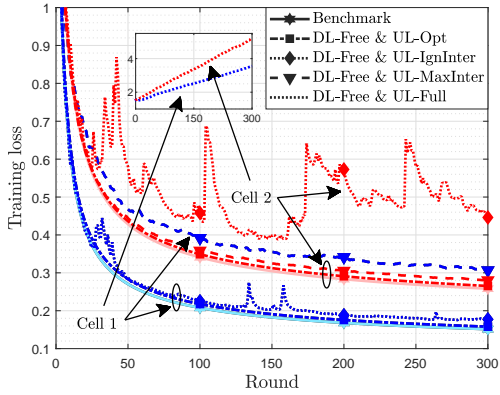}
	}
	\hspace{3mm}
	\subfigure[Test accuracy versus the number of rounds.]{
		\label{fig:ul_acc}
		\centering
		\includegraphics[scale=0.5]{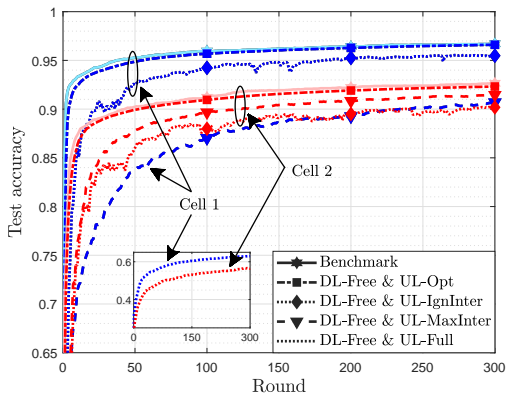}
	}
	\caption{Learning performance comparison for different uplink transmission schemes when the downlink model dissemination is error-free.} 
\end{figure}

\begin{figure}[t]
	\centering
	\subfigure[Achievable region of training loss.]{
		\label{fig:pareto_loss}
		\centering
		\includegraphics[scale=0.5]{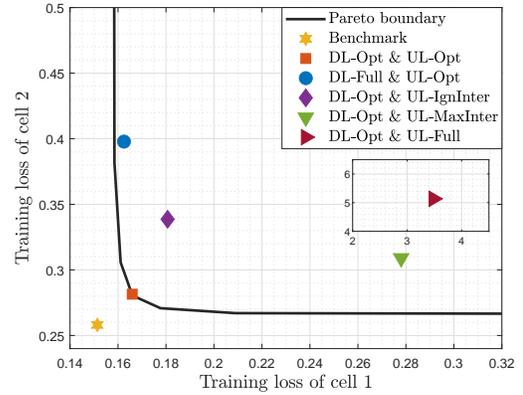}
	}
	\hspace{3mm}
	\subfigure[Achievable region of test accuracy.]{
		\label{fig:pareto_acc}
		\centering
		\includegraphics[scale=0.5]{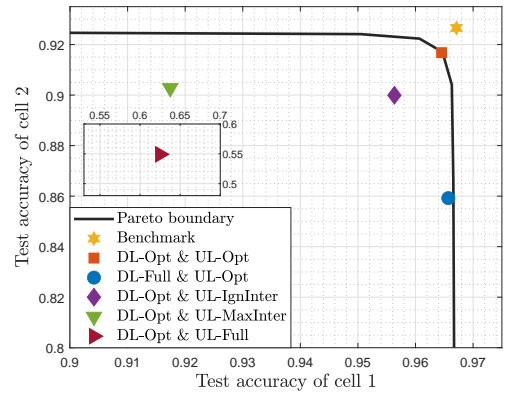}
	}
	\caption{Achievable region of learning performance when $T = 300$.}
	\label{fig:pareto_region}
\end{figure}

Figs. \ref{fig:ul_loss} and \ref{fig:ul_acc} show the learning performance versus the number of rounds for different uplink transmission schemes when the downlink model dissemination is error-free.
It is observed that \textit{UL-Opt} performs similarly to \textit{Benchmark} in terms of the training loss and test accuracy, which demonstrates that our proposed cooperative uplink transmission is capable of balancing the uplink aggregation errors between two cells, while keeping near-optimal learning performance.
Different from the downlink model dissemination, the transmission design in the uplink gradient aggregation needs to achieve the magnitude alignment at the receiver while suppressing the inter-cell interference for accomplishing the desired function computation \cite{zhu2021aircomp}.
Since \textit{UL-IgnInter} ignores the inter-cell interference and \textit{UL-MaxInter} assumes a maximum inter-cell interference, they achieve less accurate power control than \textit{UL-Opt}, which results in larger distortion of the aggregated local gradients at the BS and in turn leads to worse learning performance in both cells.
Due to the negligence of inter-cell interference, \textit{UL-IgnInter} suffers from severe inter-cell interference during the training process.
As channel conditions independently change over different communication rounds, the inter-cell interference may vary dramatically in different rounds, which leads to large fluctuations in the training loss of \textit{UL-IgnInter}.
Meanwhile, since the worst effect of inter-cell interference is considered, \textit{UL-MaxInter} yields a relatively smooth curve as compared with \textit{UL-IgnInter}.
Due to the lack of power control, the magnitude alignment of uplink local gradients achieved by \textit{UL-Full} is completely dependent on the receiver noise, channel fading, and inter-cell interference, which leads to much worse learning performance than other baseline schemes.
Meanwhile, due to the randomness of the receiver noise and channel fading, the FL training using \textit{UL-Full} is not guaranteed to converge.
This results in that some experiments yield an increasing training loss and a non-increasing test accuracy, while others yield slightly decreasing training loss and increasing test accuracy.
Consequently, by averaging the simulation results over $100$ experiments, the training loss of \textit{UL-Full} shown in Fig. \ref{fig:ul_loss} increases, instead of decreasing as other baseline schemes, while the test accuracy of \textit{UL-Full} shown in Fig. \ref{fig:ul_acc} can still increase.

Figs. \ref{fig:pareto_loss} and \ref{fig:pareto_acc} show the achievable region of learning performance when $T = 300$, where the achievable region represents the learning performance that can be simultaneously achieved by all cells under a given set of downlink and uplink transmit power constraints for BSs and devices, respectively.
The achievable regions of the training loss and test accuracy are surrounded by the Pareto boundary, where the Pareto boundary is drawn through \textit{DL-Opt \& UL-Opt} by setting the profiling vector as $\bm{\kappa}_{\rm p} = [\bar{\kappa}, 1 - \bar{\kappa}]$, where $\bar{\kappa} \in \{0.0001, 0.001, 0.1, 0.5, 0.9, 0.99, 0.9999\}$.
It is observed that all the baseline schemes that consider unreliable wireless communications lie within the achievable regions of training loss and test accuracy.
The \textit{DL-Opt \& UL-Opt} achieves the learning performance closest to \textit{Benchmark} as compared with other baseline schemes that lack interference management in either downlink or uplink transmission.
The results indicate that the interference management is required for both the downlink and uplink transmissions to balance the FL performance among multiple cells.

\subsubsection{Multi-Cell Network}

We evaluate the average learning performance over multiple cells of our proposed cooperative multi-cell FL optimization framework in multi-cell wireless networks, where cells $1 \sim M$ are active and the profiling vector is set to $\bm{\kappa} = \left[\frac{1}{M}, \frac{1}{M}, \dots, \frac{1}{M}\right]$.

\begin{figure}[t]
	\centering
	\subfigure[Average learning performance in a three-cell network.]{
		\label{fig:3cell}
		\centering
		\includegraphics[scale=0.5]{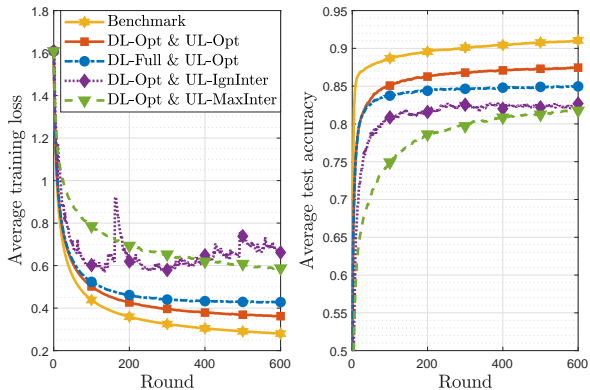}
	}
	\subfigure[Average learning performance in a four-cell network.]{
		\label{fig:4cell}
		\centering
		\includegraphics[scale=0.5]{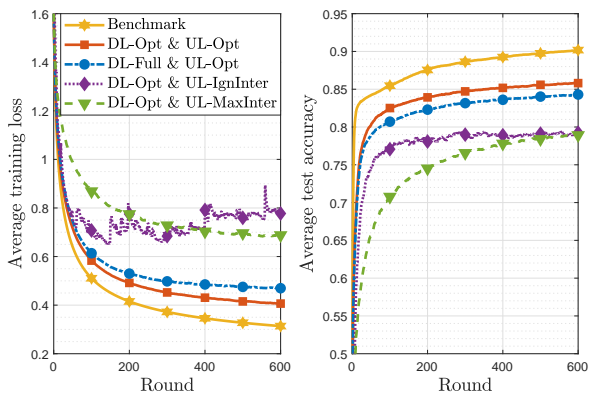}
	}
	\caption{Average learning performance comparison in multi-cell wireless networks with different transmission schemes.} 
	\label{fig:multiple_cells}
\end{figure}

Figs. \ref{fig:3cell} and \ref{fig:4cell} show the average learning performance versus the number of rounds under different transmission schemes in three-cell and four-cell networks, respectively.
It is observed that \textit{DL-Opt \& UL-Opt} outperforms other baseline schemes that lack interference management in either downlink or uplink transmission in terms of the average training loss and test accuracy.
This demonstrates that our proposed cooperative transmission design is effective in a multi-cell wireless network and can balance the learning performance.
As the number of rounds increases, the learning performance of \textit{DL-Opt \& UL-MaxInter} gradually approaches and even exceeds the learning performance of \textit{DL-Opt \& UL-IgnInter}.
Hence, the inter-cell interference gradually becomes a major factor limiting the performance improvement, which demonstrates the importance of interference management in multi-cell wireless networks.
Also, since the performance of \textit{DL-Opt \& UL-IgnInter} is mainly limited by the uplink transmission, a phenomenon similar to \textit{DL-Free \& UL-IgnInter} is also observed in Fig. \ref{fig:multiple_cells}, i.e., ignoring the inter-cell interference in the upink transmission leads to relatively large fluctuations in the training loss.

\section{Conclusions} \label{Sec:Conc}

In this paper, we consider over-the-air FL in a multi-cell wireless network, where each cell performs a different FL task.
To quantify the learning performance, we first conduct the convergence analysis of AirComp-assisted FL systems, which reveals that the distorted downlink and uplink communications result in a gap that hinders the convergence of FL.
We then characterize the Pareto boundary of the gap region via the profiling technique and further formulate an optimization problem to minimize the sum of error-induced gaps for all cells.
Subsequently, we propose a cooperative multi-cell FL optimization framework to achieve efficient interference management in the downlink and uplink transmissions.
Benefiting from the coordination among multiple cells, our proposed algorithm is able to achieve much better average learning performance than non-cooperative baseline schemes.
 
\appendix[Proof of Theorem \ref{theorem:convergence}] \label{appendix:convergence}

For presentation clarity, we omit the cell index in the following analysis.
According to \eqref{eq:dl_w}, \eqref{eq:ul_g}, and \eqref{eq:global_update}, the single-round change of the global model is given by
\begin{align} \label{eq:w_diff}
	& \bm{w}^{t + 1} - \bm{w}^t = - \frac{\eta^{t + 1}}{K} \left(\sum_{k \in \mathcal{K}} \bm{g}_k^{t + 1} + \underset{\bm{e}_{\rm u}^{t + 1}}{\underbrace{\Re\left\{(\bm{e}^{\rm ul})^{t + 1}\right\}}}\right) \notag \\
	& = - \frac{\eta^{t + 1}}{K} \left(\sum_{k \in \mathcal{K}} \nabla F_k(\bm{w}^t + \underset{\bm{e}_{{\rm d}, k}^{t + 1}}{\underbrace{\Re\left\{(\bm{e}_k^{\rm dl})^{t + 1}\right\}}}) + \bm{e}_{\rm u}^{t + 1}\right).
\end{align}

Based on \eqref{eq:global_loss}, \eqref{eq:smooth}, and \eqref{eq:w_diff}, we have
\begin{align} \label{eq:F_diff}
	& F(\bm{w}^{t + 1}) - F(\bm{w}^t) \notag \\
	& \le \left<\nabla F(\bm{w}^t), \bm{w}^{t + 1} - \bm{w}^t\right> + \frac{L}{2} \left\|\bm{w}^{t + 1} - \bm{w}^t\right\|^2 \notag \\
	& = - \eta^{t + 1} \left<\nabla F(\bm{w}^t), \frac{1}{K} \sum_{k \in \mathcal{K}} \nabla F_k(\bm{w}^t + \bm{e}_{{\rm d}, k}^{t + 1})\right> \notag \\
	& \quad \underset{A_1}{\underbrace{- \frac{\eta^{t + 1}}{K} \left<\nabla F(\bm{w}^t), \bm{e}_{\rm u}^{t + 1}\right>}} \notag \\
	& \quad + \frac{L (\eta^{t + 1})^2}{2}\underset{A_2}{\underbrace{\left\|\frac{1}{K} \sum_{k \in \mathcal{K}} \nabla F_k(\bm{w}^t + \bm{e}_{{\rm d}, k}^{t + 1}) + \frac{1}{K} \bm{e}_{\rm u}^{t + 1}\right\|^2}}.
\end{align}
We then rewrite $A_2$ as
\begin{align} \label{eq:A_2}
	& A_2 = \left\|\frac{1}{K} \sum_{k \in \mathcal{K}} \nabla F_k(\bm{w}^t + \bm{e}_{{\rm d}, k}^{t + 1}) + \frac{1}{K} \bm{e}_{\rm u}^{t + 1}\right\|^2 \notag \\
	& = \left\|\frac{1}{K} \sum_{k \in \mathcal{K}} \nabla F_k(\bm{w}^t + \bm{e}_{{\rm d}, k}^{t + 1})\right\|^2 + \frac{1}{K^2} \left\|\bm{e}_{\rm u}^{t + 1}\right\|^2 \notag \\
	& \quad + 2 \left<\frac{1}{K} \sum_{k \in \mathcal{K}} \nabla F_k(\bm{w}^t + \bm{e}_{{\rm d}, k}^{t + 1}), \frac{1}{K} \bm{e}_{\rm u}^{t + 1}\right> \notag \\
	& \overset{\diamondsuit_1}{=} \frac{1}{K^2} \left\|\bm{e}_{\rm u}^{t + 1}\right\|^2 + 2 \left<\frac{1}{K} \sum_{k \in \mathcal{K}} \nabla F_k(\bm{w}^t + \bm{e}_{{\rm d}, k}^{t + 1}), \frac{1}{K} \bm{e}_{\rm u}^{t + 1}\right> \notag \\
	& \quad + \left\|\frac{1}{K} \sum_{k \in \mathcal{K}} \nabla F_k(\bm{w}^t + \bm{e}_{{\rm d}, k}^{t + 1}) - \nabla F(\bm{w}^t)\right\|^2 - \left\|\nabla F(\bm{w}^t)\right\|^2 \notag \\
	& \quad + 2 \left<\nabla F(\bm{w}^t), \frac{1}{K} \sum_{k \in \mathcal{K}} \nabla F_k(\bm{w}^t + \bm{e}_{{\rm d}, k}^{t + 1})\right>
\end{align}
where $\diamondsuit_1$ is according to
\begin{align}
	& \left\|\frac{1}{K} \sum_{k \in \mathcal{K}} \nabla F_k(\bm{w}^t + \bm{e}_{{\rm d}, k}^{t + 1})\right\|^2 \notag \\
	& = \left\|\frac{1}{K} \sum_{k \in \mathcal{K}} \nabla F_k(\bm{w}^t + \bm{e}_{{\rm d}, k}^{t + 1}) - \nabla F(\bm{w}^t) + \nabla F(\bm{w}^t)\right\|^2 \notag \\
	& = \left\|\frac{1}{K} \sum_{k \in \mathcal{K}} \nabla F_k(\bm{w}^t + \bm{e}_{{\rm d}, k}^{t + 1}) - \nabla F(\bm{w}^t)\right\|^2 + \left\|\nabla F(\bm{w}^t)\right\|^2 \notag \\
	& \quad + 2 \left<\nabla F(\bm{w}^t), \frac{1}{K} \sum_{k \in \mathcal{K}} \nabla F_k(\bm{w}^t + \bm{e}_{{\rm d}, k}^{t + 1}) - \nabla F(\bm{w}^t)\right> \notag \\
	& = \left\|\frac{1}{K} \sum_{k \in \mathcal{K}} \nabla F_k(\bm{w}^t + \bm{e}_{{\rm d}, k}^{t + 1}) - \nabla F(\bm{w}^t)\right\|^2 - \left\|\nabla F(\bm{w}^t)\right\|^2 \notag \\
	& \quad + 2 \left<\nabla F(\bm{w}^t), \frac{1}{K} \sum_{k \in \mathcal{K}} \nabla F_k(\bm{w}^t + \bm{e}_{{\rm d}, k}^{t + 1})\right>.
\end{align}

By setting $0 < \eta^t \equiv \eta < \frac{1}{L}$ and substituting \eqref{eq:A_2} into \eqref{eq:F_diff}, we have
\begin{align} \label{eq:F_diff_2}
	& F(\bm{w}^{t + 1}) - F(\bm{w}^t) \notag \\
	& \le \underset{A_3}{\underbrace{A_1 + L \eta^2 \left<\frac{1}{K} \sum_{k \in \mathcal{K}} \nabla F_k(\bm{w}^t + \bm{e}_{{\rm d}, k}^{t + 1}), \frac{1}{K} \bm{e}_{\rm u}^{t + 1}\right>}} \notag \\
	& \quad + \frac{L \eta^2}{2 K^2} \left\|\bm{e}_{\rm u}^{t + 1}\right\|^2 - \frac{L \eta^2}{2} \left\|\nabla F(\bm{w}^t)\right\|^2 \notag \\
	& \quad + \frac{L \eta^2}{2} \left\|\frac{1}{K} \sum_{k \in \mathcal{K}} \nabla F_k(\bm{w}^t + \bm{e}_{{\rm d}, k}^{t + 1}) - \nabla F(\bm{w}^t)\right\|^2 \notag \\
	& \quad - \eta \left(1 - L \eta\right) \left<\nabla F(\bm{w}^t), \frac{1}{K} \sum_{k \in \mathcal{K}} \nabla F_k(\bm{w}^t + \bm{e}_{{\rm d}, k}^{t + 1})\right> \notag \\
	& \overset{\diamondsuit_2}{=} A_3 + \frac{L \eta^2}{2 K^2} \left\|\bm{e}_{\rm u}^{t + 1}\right\|^2 - \eta \left(1 - L \eta\right) \notag \\
	& \quad \times \left<\nabla F(\bm{w}^t), \frac{1}{K} \sum_{k \in \mathcal{K}} \nabla F_k(\bm{w}^t + \bm{e}_{{\rm d}, k}^{t + 1}) - \nabla F(\bm{w}^t)\right> \notag \\
	& \quad - \eta\left(1 - \frac{L \eta}{2}\right) \left\|\nabla F(\bm{w}^t)\right\|^2 \notag \\
	& \quad + \frac{L \eta^2}{2} \left\|\frac{1}{K} \sum_{k \in \mathcal{K}} \nabla F_k(\bm{w}^t + \bm{e}_{{\rm d}, k}^{t + 1}) - \nabla F(\bm{w}^t)\right\|^2 \notag \\
	& \overset{\diamondsuit_3}{\le} A_3 + \frac{L \eta^2}{2 K^2} \left\|\bm{e}_{\rm u}^{t + 1}\right\|^2 - \frac{\eta}{2} \left\|\nabla F(\bm{w}^t)\right\|^2 \notag \\
	& \quad + \frac{\eta}{2} \left\|\frac{1}{K} \sum_{k \in \mathcal{K}} \nabla F_k(\bm{w}^t + \bm{e}_{{\rm d}, k}^{t + 1}) - \nabla F(\bm{w}^t)\right\|^2
\end{align}
where $\diamondsuit_2$ holds by setting
\begin{align}
	& \left<\nabla F(\bm{w}^t), \frac{1}{K} \sum_{k \in \mathcal{K}} \nabla F_k(\bm{w}^t + \bm{e}_{{\rm d}, k}^{t + 1})\right> \notag \\
	& = \left<\nabla F(\bm{w}^t), \nabla F(\bm{w}^t) \right. \notag \\
	& \quad + \left. \frac{1}{K} \sum_{k \in \mathcal{K}} \nabla F_k(\bm{w}^t + \bm{e}_{{\rm d}, k}^{t + 1}) - \nabla F(\bm{w}^t)\right>
\end{align}
and $\diamondsuit_3$ follows from
\begin{align}
	& - 2 \left<\nabla F(\bm{w}^t), \frac{1}{K} \sum_{k \in \mathcal{K}} \nabla F_k(\bm{w}^t + \bm{e}_{{\rm d}, k}^{t + 1}) - \nabla F(\bm{w}^t)\right> \notag \\
	& \le \left\|\nabla F(\bm{w}^t)\right\|^2 + \left\|\frac{1}{K} \sum_{k \in \mathcal{K}} \nabla F_k(\bm{w}^t + \bm{e}_{{\rm d}, k}^{t + 1}) - \nabla F(\bm{w}^t)\right\|^2.
\end{align}

Since model/gradient parameter $\bm{\theta}^t \in \{\bm{w}_m^t, \bm{g}^t_k\}$ at the $t$-th round is determined by the normalized uplink/downlink transmit symbols, $\mathcal{S} = \bigcup_{i = 1}^{t - 1} \left\{(\bm{s}_m^{\rm dl})^i, (\bm{s}_k^{\rm ul})^i\right\}$, and the receiver noises, $\mathcal{Z} = \bigcup_{i = 1}^{t - 1} \left\{(\bm{z}_k^{\rm dl})^i, (\bm{z}_m^{\rm ul})^i\right\}$, in the previous $(t - 1)$ rounds, the total expectation of $\bm{\theta}^t$ and $F(\bm{\theta}^t)$ for any $t \in \mathcal{T}$ can be represented as $\mathbb{E}\left[\bm{\theta}^t\right] = \mathbb{E}_\mathcal{S} \mathbb{E}_\mathcal{Z}\left[\bm{\theta}^t\right]$ and $\mathbb{E}\left[F(\bm{\theta}^t)\right] = \mathbb{E}_\mathcal{S} \mathbb{E}_\mathcal{Z}\left[F(\bm{\theta}^t)\right]$, respectively \cite{friedlander2012hybrid}.
In addition, since $\mathbb{E}\left[\bm{\phi}\right] = \bm{0}$, $\forall \, \bm{\phi} \in \mathcal{S} \cup \mathcal{Z}$, we have
$\mathbb{E}\left[\bm{e}_{\rm u}^t\right] = \bm{0}$ and $\mathbb{E}\left[\bm{e}_{{\rm d}, k}^t\right] = \bm{0}$, $\forall \, k \in \mathcal{K}$, which implies that $\mathbb{E}\left[A_3\right] = 0$ for $A_3$ given in \eqref{eq:F_diff_2}.

By summing both sides of \eqref{eq:F_diff_2} for $T$ rounds and taking the total expectation, based on Assumption \ref{assump:lower_bound}, we have
\begin{align} \label{eq:F_diff_3}
	& F(\bm{w}^\star) - F(\bm{w}^0) \le \mathbb{E}\left[F(\bm{w}^T)\right] - F(\bm{w}^0) \notag \\
	& \le - \frac{\eta}{2} \sum_{t = 0}^{T - 1} \mathbb{E}\left[\left\|\nabla F(\bm{w}^t)\right\|^2\right] + \frac{L \eta^2}{2 K^2} \sum_{t = 0}^{T - 1} \mathbb{E}\left[\left\|\bm{e}_{\rm u}^{t + 1}\right\|^2\right] \notag \\
	& \quad + \frac{\eta}{2} \sum_{t = 0}^{T - 1} \mathbb{E}\left[\left\|\frac{1}{K} \sum_{k \in \mathcal{K}} \nabla F_k(\bm{w}^t + \bm{e}_{{\rm d}, k}^{t + 1}) - \nabla F(\bm{w}^t)\right\|^2\right] \notag \\
	& \overset{\diamondsuit_4}{\le} - \frac{\eta}{2} \sum_{t = 0}^{T - 1} \mathbb{E}\left[\left\|\nabla F(\bm{w}^t)\right\|^2\right] + \frac{L \eta^2}{2 K^2} \sum_{t = 0}^{T - 1} \mathbb{E}\left[\left\|\bm{e}_{\rm u}^{t + 1}\right\|^2\right] \notag \\
	& \quad + \frac{\eta}{2 K} \sum_{t = 0}^{T - 1} \sum_{k \in \mathcal{K}} \mathbb{E}\left[\left\|\nabla F_k(\bm{w}^t + \bm{e}_{{\rm d}, k}^{t + 1}) - \nabla F_k(\bm{w}^t)\right\|^2\right] \notag \\
	& \overset{\diamondsuit_5}{\le} - \frac{\eta}{2} \sum_{t = 0}^{T - 1} \mathbb{E}\left[\left\|\nabla F(\bm{w}^t)\right\|^2\right] + \frac{L \eta^2}{2 K^2} \sum_{t = 0}^{T - 1} \mathbb{E}\left[\left\|\bm{e}_{\rm u}^{t + 1}\right\|^2\right] \notag \\
	& \quad + \frac{\eta L^2}{2 K} \sum_{t = 0}^{T - 1} \sum_{k \in \mathcal{K}} \mathbb{E}\left[\left\|\bm{e}_{{\rm d}, k}^{t + 1}\right\|^2\right]
\end{align}
where $\diamondsuit_4$ follows from the Jensen inequality and $\diamondsuit_5$ is based on \eqref{eq:lipschitz} in Assumption \ref{assump:smooth}.
Finally, dividing both sides of \eqref{eq:F_diff_3} by $T$ rounds and rearranging it yield \eqref{eq:cvg}.

\bibliographystyle{IEEEtran}
\bibliography{ref}

\begin{thebibliography}{10}
\providecommand{\url}[1]{#1}
\csname url@samestyle\endcsname
\providecommand{\newblock}{\relax}
\providecommand{\bibinfo}[2]{#2}
\providecommand{\BIBentrySTDinterwordspacing}{\spaceskip=0pt\relax}
\providecommand{\BIBentryALTinterwordstretchfactor}{4}
\providecommand{\BIBentryALTinterwordspacing}{\spaceskip=\fontdimen2\font plus
\BIBentryALTinterwordstretchfactor\fontdimen3\font minus
  \fontdimen4\font\relax}
\providecommand{\BIBforeignlanguage}[2]{{%
\expandafter\ifx\csname l@#1\endcsname\relax
\typeout{** WARNING: IEEEtran.bst: No hyphenation pattern has been}%
\typeout{** loaded for the language `#1'. Using the pattern for}%
\typeout{** the default language instead.}%
\else
\language=\csname l@#1\endcsname
\fi
#2}}
\providecommand{\BIBdecl}{\relax}
\BIBdecl

\bibitem{letaief2019roadmap}
K.~B. Letaief, W.~Chen, Y.~Shi, J.~Zhang, and Y.-J.~A. Zhang, ``The roadmap to
  {6G}: {AI} empowered wireless networks,'' \emph{IEEE Commun. Mag.}, vol.~57,
  no.~8, pp. 84--90, Aug. 2019.

\bibitem{yang2020artificial}
H.~Yang, A.~Alphones, Z.~Xiong, D.~Niyato, J.~Zhao, and K.~Wu,
  ``Artificial-intelligence-enabled intelligent {6G} networks,'' \emph{IEEE
  Netw.}, Nov./Dec. 2020.

\bibitem{shi2021mobile}
Y.~Shi, K.~Yang, Z.~Yang, and Y.~Zhou, \emph{Mobile Edge Artificial
  Intelligence: Opportunities and Challenges}.\hskip 1em plus 0.5em minus
  0.4em\relax Elsevier, Aug. 2021.

\bibitem{konevcny2016federated}
\BIBentryALTinterwordspacing
J.~Kone{\v{c}}n{\`y}, H.~B. McMahan, F.~X. Yu, P.~Richt{\'a}rik, A.~T. Suresh,
  and D.~Bacon, ``Federated learning: Strategies for improving communication
  efficiency,'' 2016. [Online]. Available:
  \url{http://arxiv.org/abs/1610.05492}
\BIBentrySTDinterwordspacing

\bibitem{mcmahan2017communication}
B.~McMahan, E.~Moore, D.~Ramage, S.~Hampson, and B.~A. y~Arcas,
  ``Communication-efficient learning of deep networks from decentralized
  data,'' in \emph{Proc. Int. Conf. Artif. Intell. Stat. (AISTATS)}, Apr. 2017,
  pp. 1273--1282.

\bibitem{wang2019adaptive}
S.~Wang, T.~Tuor, T.~Salonidis, K.~K. Leung, C.~Makaya, T.~He, and K.~Chan,
  ``Adaptive federated learning in resource constrained edge computing
  systems,'' \emph{IEEE J. Sel. Areas Commun.}, vol.~37, no.~6, pp. 1205--1221,
  Jun. 2019.

\bibitem{dinh2021federated}
C.~T. Dinh, N.~H. Tran, M.~N.~H. Nguyen, C.~S. Hong, W.~Bao, A.~Y. Zomaya, and
  V.~Gramoli, ``Federated learning over wireless networks: Convergence analysis
  and resource allocation,'' \emph{IEEE/ACM Trans. Netw.}, vol.~29, no.~1, pp.
  398--409, Feb. 2021.

\bibitem{yang2020scheduling}
H.~H. Yang, Z.~Liu, T.~Q.~S. Quek, and H.~V. Poor, ``Scheduling policies for
  federated learning in wireless networks,'' \emph{IEEE Trans. Commun.},
  vol.~68, no.~1, pp. 317--333, Jan. 2020.

\bibitem{kang2020reliable}
J.~Kang, Z.~Xiong, D.~Niyato, Y.~Zou, Y.~Zhang, and M.~Guizani, ``Reliable
  federated learning for mobile networks,'' \emph{IEEE Wireless Commun.},
  vol.~27, no.~2, pp. 72--80, Apr. 2020.

\bibitem{amiri2021convergence}
M.~M. Amiri, D.~Gündüz, S.~R. Kulkarni, and H.~V. Poor, ``Convergence of
  update aware device scheduling for federated learning at the wireless edge,''
  \emph{IEEE Trans. Wireless Commun.}, vol.~20, no.~6, pp. 3643--3658, Jun.
  2021.

\bibitem{ren2020scheduling}
J.~Ren, Y.~He, D.~Wen, G.~Yu, K.~Huang, and D.~Guo, ``Scheduling for cellular
  federated edge learning with importance and channel awareness,'' \emph{IEEE
  Trans. Wireless Commun.}, vol.~19, no.~11, pp. 7690--7703, Nov. 2020.

\bibitem{chen2021joint}
M.~Chen, Z.~Yang, W.~Saad, C.~Yin, H.~V. Poor, and S.~Cui, ``A joint learning
  and communications framework for federated learning over wireless networks,''
  \emph{IEEE Trans. Wireless Commun.}, vol.~20, no.~1, pp. 269--283, Jan. 2021.

\bibitem{shi2021joint}
W.~Shi, S.~Zhou, Z.~Niu, M.~Jiang, and L.~Geng, ``Joint device scheduling and
  resource allocation for latency constrained wireless federated learning,''
  \emph{IEEE Trans. Wireless Commun.}, vol.~20, no.~1, pp. 453--467, Jan. 2021.

\bibitem{wei2022lowlatency}
K.~Wei, J.~Li, C.~Ma, M.~Ding, C.~Chen, S.~Jin, Z.~Han, and H.~V. Poor,
  ``Low-latency federated learning over wireless channels with differential
  privacy,'' \emph{IEEE J. Sel. Areas Commun.}, vol.~40, no.~1, pp. 290--307,
  Jan. 2022.

\bibitem{xu2021client}
J.~Xu and H.~Wang, ``Client selection and bandwidth allocation in wireless
  federated learning networks: A long-term perspective,'' \emph{IEEE Trans.
  Wireless Commun.}, vol.~20, no.~2, pp. 1188--1200, Feb. 2021.

\bibitem{abad2020hierarchical}
M.~S.~H. Abad, E.~Ozfatura, D.~Gündüz, and O.~Ercetin, ``Hierarchical
  federated learning across heterogeneous cellular networks,'' in \emph{Proc.
  IEEE Int. Conf. Acoust. Speech Signal Process. (ICASSP)}, 2020, pp.
  8866--8870.

\bibitem{luo2020hfel}
S.~Luo, X.~Chen, Q.~Wu, Z.~Zhou, and S.~Yu, ``{HFEL}: Joint edge association
  and resource allocation for cost-efficient hierarchical federated edge
  learning,'' \emph{IEEE Trans. Wireless Commun.}, vol.~19, no.~10, pp.
  6535--6548, May 2020.

\bibitem{lim2021dynamic}
W.~Y.~B. Lim, J.~S. Ng, Z.~Xiong, D.~Niyato, C.~Miao, and D.~I. Kim, ``Dynamic
  edge association and resource allocation in self-organizing hierarchical
  federated learning networks,'' \emph{IEEE J. Sel. Areas Commun.}, vol.~39,
  no.~12, pp. 3640--3653, Dec. 2021.

\bibitem{lim2022decentralized}
W.~Y.~B. Lim, J.~S. Ng, Z.~Xiong, J.~Jin, Y.~Zhang, D.~Niyato, C.~Leung, and
  C.~Miao, ``Decentralized edge intelligence: A dynamic resource allocation
  framework for hierarchical federated learning,'' \emph{IEEE Trans. Parallel
  Distrib. Syst.}, vol.~33, no.~3, pp. 536--550, Mar. 2022.

\bibitem{yang2020federated}
K.~{Yang}, T.~{Jiang}, Y.~{Shi}, and Z.~{Ding}, ``Federated learning via
  over-the-air computation,'' \emph{IEEE Trans. Wireless Commun.}, vol.~19,
  no.~3, pp. 2022--2035, Mar. 2020.

\bibitem{zhu2020broadband}
G.~Zhu, Y.~Wang, and K.~Huang, ``Broadband analog aggregation for low-latency
  federated edge learning,'' \emph{IEEE Trans. Wireless Commun.}, vol.~19,
  no.~1, pp. 491--506, Jan. 2020.

\bibitem{amiri2020machine}
M.~M. Amiri and D.~Gündüz, ``Machine learning at the wireless edge:
  Distributed stochastic gradient descent over-the-air,'' \emph{IEEE Trans.
  Signal Process.}, vol.~68, pp. 2155--2169, Mar. 2020.

\bibitem{fang2022communicationefficient}
\BIBentryALTinterwordspacing
W.~Fang, Z.~Yu, Y.~Jiang, Y.~Shi, C.~N. Jones, and Y.~Zhou,
  ``Communication-efficient stochastic zeroth-order optimization for federated
  learning,'' 2022. [Online]. Available: \url{http://arxiv.org/abs/2201.09531}
\BIBentrySTDinterwordspacing

\bibitem{zhu2021aircomp}
G.~Zhu, J.~Xu, K.~Huang, and S.~Cui, ``Over-the-air computing for wireless data
  aggregation in massive {IoT},'' \emph{IEEE Wireless Commun.}, vol.~28, no.~4,
  pp. 57--65, Aug. 2021.

\bibitem{wang2021wirelesspowered}
Z.~Wang, Y.~Shi, Y.~Zhou, H.~Zhou, and N.~Zhang, ``Wireless-powered
  over-the-air computation in intelligent reflecting surface-aided {IoT}
  networks,'' \emph{IEEE Internet Things J.}, vol.~8, no.~3, pp. 1585--1598,
  Feb. 2021.

\bibitem{fang2021overtheair}
W.~Fang, Y.~Jiang, Y.~Shi, Y.~Zhou, W.~Chen, and K.~B. Letaief, ``Over-the-air
  computation via reconfigurable intelligent surface,'' \emph{IEEE Trans.
  Commun.}, vol.~69, no.~12, pp. 8612--8626, Dec. 2021.

\bibitem{zhang2021gradient}
N.~Zhang and M.~Tao, ``Gradient statistics aware power control for over-the-air
  federated learning,'' \emph{IEEE Trans. Wireless Commun.}, vol.~20, no.~8,
  pp. 5115--5128, Aug. 2021.

\bibitem{cao2021optimized}
X.~Cao, G.~Zhu, J.~Xu, Z.~Wang, and S.~Cui, ``Optimized power control design
  for over-the-air federated edge learning,'' \emph{IEEE J. Sel. Areas
  Commun.}, vol.~40, no.~1, pp. 342--358, Jan. 2022.

\bibitem{yang2020fedirs}
K.~Yang, Y.~Shi, Y.~Zhou, Z.~Yang, L.~Fu, and W.~Chen, ``Federated machine
  learning for intelligent {IoT} via reconfigurable intelligent surface,''
  \emph{IEEE Netw.}, Sep./Oct. 2020.

\bibitem{liu2021reconfigurable}
H.~Liu, X.~Yuan, and Y.-J.~A. Zhang, ``Reconfigurable intelligent surface
  enabled federated learning: A unified communication-learning design
  approach,'' \emph{IEEE Trans. Wireless Commun.}, vol.~20, no.~11, pp.
  7595--7609, Nov. 2021.

\bibitem{wang2022federated}
Z.~Wang, J.~Qiu, Y.~Zhou, Y.~Shi, L.~Fu, W.~Chen, and K.~B. Letaief,
  ``Federated learning via intelligent reflecting surface,'' \emph{IEEE Trans.
  Wireless Commun.}, vol.~21, no.~2, pp. 808--822, Feb. 2022.

\bibitem{lin2022relayassisted}
Z.~Lin, H.~Liu, and Y.-J.~A. Zhang, ``Relay-assisted cooperative federated
  learning,'' \emph{IEEE Trans. Wireless Commun.}, 2022, Early Access, doi:
  \href{https://doi.org/10.1109/TWC.2022.3155596}{10.1109/TWC.2022.3155596}.

\bibitem{amiri2021noisydl}
M.~M. Amiri, D.~Gündüz, S.~R. Kulkarni, and H.~V. Poor, ``Convergence of
  federated learning over a noisy downlink,'' \emph{IEEE Trans. Wireless
  Commun.}, vol.~21, no.~3, pp. 1422--1437, Mar. 2022.

\bibitem{wei2022federated}
X.~Wei and C.~Shen, ``Federated learning over noisy channels: Convergence
  analysis and design examples,'' \emph{IEEE Trans. Cogn. Commun. Netw.}, 2022,
  Early Access, doi:
  \href{https://doi.org/10.1109/TCCN.2022.3140788}{10.1109/TCCN.2022.3140788}.

\bibitem{wang2022edge}
S.~Wang, Y.~Hong, R.~Wang, Q.~Hao, Y.-C. Wu, and D.~W.~K. Ng, ``Edge federated
  learning via unit-modulus over-the-air computation,'' \emph{IEEE Trans.
  Commun.}, 2022, Early Access, doi:
  \href{https://doi.org/10.1109/TCOMM.2022.3153488}{10.1109/TCOMM.2022.3153488}.

\bibitem{salehi2021federated}
M.~Salehi and E.~Hossain, ``Federated learning in unreliable and
  resource-constrained cellular wireless networks,'' \emph{IEEE Trans.
  Commun.}, vol.~69, no.~8, pp. 5136--5151, Aug. 2021.

\bibitem{lin2021deploying}
Z.~Lin, X.~Li, V.~K.~N. Lau, Y.~Gong, and K.~Huang, ``Deploying federated
  learning in large-scale cellular networks: Spatial convergence analysis,''
  \emph{IEEE Trans. Wireless Commun.}, vol.~21, no.~3, pp. 1542--1556, Mar.
  2022.

\bibitem{cao2021cooperative}
X.~Cao, G.~Zhu, J.~Xu, and K.~Huang, ``Cooperative interference management for
  over-the-air computation networks,'' \emph{IEEE Trans. Wireless Commun.},
  vol.~20, no.~4, pp. 2634--2651, Apr. 2021.

\bibitem{amiri2020federated}
M.~M. Amiri and D.~Gündüz, ``Federated learning over wireless fading
  channels,'' \emph{IEEE Trans. Wireless Commun.}, vol.~19, no.~5, pp.
  3546--3557, May 2020.

\bibitem{abari2015airshare}
O.~Abari, H.~Rahul, D.~Katabi, and M.~Pant, ``{AirShare}: Distributed coherent
  transmission made seamless,'' in \emph{Proc. IEEE Conf. Comput. Commun.
  (INFOCOM)}, Apr. 2015, pp. 1742--1750.

\bibitem{mahmood2019time}
A.~Mahmood, M.~I. Ashraf, M.~Gidlund, J.~Torsner, and J.~Sachs, ``Time
  synchronization in {5G} wireless edge: Requirements and solutions for
  critical-{MTC},'' \emph{IEEE Commun. Mag.}, vol.~57, no.~12, pp. 45--51, Dec.
  2019.

\bibitem{allen2018natasha}
Z.~Allen-Zhu, ``Natasha 2: Faster non-convex optimization than {SGD},'' in
  \emph{Proc. Neural Inf. Process. Syst. (NeurIPS)}, Dec. 2018.

\bibitem{bottou2018optimization}
L.~Bottou, F.~Curtis, and J.~Nocedal, ``Optimization methods for large-scale
  machine learning,'' \emph{SIAM Review}, vol.~60, no.~2, pp. 223--311, May
  2018.

\bibitem{friedlander2012hybrid}
M.~P. Friedlander and M.~Schmidt, ``Hybrid deterministic-stochastic methods for
  data fitting,'' \emph{SIAM J. Sci. Comput.}, vol.~34, no.~3, pp.
  A1380--A1405, May 2012.

\bibitem{jorswieck2008complete}
E.~A. Jorswieck, E.~G. Larsson, and D.~Danev, ``Complete characterization of
  the {Pareto} boundary for the {MISO} interference channel,'' \emph{IEEE
  Trans. Signal Process.}, vol.~56, no.~10, pp. 5292--5296, Oct. 2008.

\bibitem{zhang2010cooperative}
R.~Zhang and S.~Cui, ``Cooperative interference management with {MISO}
  beamforming,'' \emph{IEEE Trans. Signal Process.}, vol.~58, no.~10, pp.
  5450--5458, Oct. 2010.

\bibitem{cvx}
M.~Grant and S.~Boyd, ``{CVX}: Matlab software for disciplined convex
  programming, version 2.1,'' \url{http://cvxr.com/cvx}, 2014.

\bibitem{polik2010interior}
I.~P{\'o}lik and T.~Terlaky, ``Interior point methods for nonlinear
  optimization,'' in \emph{Nonlinear optimization}.\hskip 1em plus 0.5em minus
  0.4em\relax Springer, Jan. 2010, pp. 215--276.

\bibitem{beko2012efficient}
M.~Beko, ``Efficient beamforming in cognitive radio multicast transmission,''
  \emph{IEEE Trans. Wireless Commun.}, vol.~11, no.~11, pp. 4108--4117, Nov.
  2012.

\bibitem{zarikoff2010coordinated}
B.~W. Zarikoff and J.~K. Cavers, ``Coordinated multi-cell systems: Carrier
  frequency offset estimation and correction,'' \emph{IEEE J. Sel. Areas
  Commun.}, vol.~28, no.~9, pp. 1490--1501, Dec. 2010.

\bibitem{tse2005fundamentals}
D.~Tse and P.~Viswanath, \emph{Fundamentals of wireless communication}.\hskip
  1em plus 0.5em minus 0.4em\relax Cambridge Univ. Press, Jul. 2005.

\bibitem{jurafsky2009speech}
D.~Jurafsky and J.~H. Martin, \emph{Speech and Language Processing (2nd
  Edition)}.\hskip 1em plus 0.5em minus 0.4em\relax Prentice Hall, May 2008.

\bibitem{lecun1998gradient}
Y.~Lecun, L.~Bottou, Y.~Bengio, and P.~Haffner, ``Gradient-based learning
  applied to document recognition,'' \emph{Proc. IEEE}, vol.~86, no.~11, pp.
  2278--2324, Nov. 1998.

\bibitem{xiao2017fashion}
\BIBentryALTinterwordspacing
H.~Xiao, K.~Rasul, and R.~Vollgraf, ``Fashion-{MNIST}: a novel image dataset
  for benchmarking machine learning algorithms,'' 2017. [Online]. Available:
  \url{http://arxiv.org/abs/1708.07747}
\BIBentrySTDinterwordspacing

\bibitem{cao2020optimized}
X.~Cao, G.~Zhu, J.~Xu, and K.~Huang, ``Optimized power control for over-the-air
  computation in fading channels,'' \emph{IEEE Trans. Wireless Commun.},
  vol.~19, no.~11, pp. 7498--7513, Nov. 2020.

\end{thebibliography}

\end{document}